\documentclass[a4paper,reqno]{amsart}
\pdfoutput=1

\usepackage{microtype}
\usepackage{graphicx}

\usepackage{mathtools}
\usepackage{amsthm}
\usepackage{amssymb}
\usepackage{amsmath}
\usepackage{stmaryrd}

\usepackage[utf8]{inputenc}
\usepackage[T1]{fontenc}
\usepackage{lmodern}

\usepackage[dvipsnames]{xcolor}
\usepackage{tikz-cd}
\usepackage{enumerate}
\renewcommand{\eqref}[1]{(\ref{#1})}

\usepackage[hyphens]{url}
\usepackage[
  colorlinks=true  
, linktocpage=true 
, linkcolor=NavyBlue    
, citecolor=Green  
, urlcolor=BrickRed 
]{hyperref}
\usepackage[capitalise,noabbrev]{cleveref}

\usepackage[backend=biber,style=alphabetic,useprefix=true]{biblatex}
\newcommand{\mkTTurl}[1]{\href{https://www.cs.bham.ac.uk/~mhe/TypeTopology/Published.#1.html}{\texttt{#1}}}

\newcommand{\ALF}{\textsc{ALF}}
\newcommand{\Agda}{\textsc{Agda}}
\newcommand{\Coq}{\textsc{Coq}}
\newcommand{\TypeTopology}{\textsc{TypeTopology}}

\newcommand{\colonequiv}{\mathrel{\vcentcolon\mspace{-1.2mu}\equiv}}

\DeclarePairedDelimiter{\pa}{(}{)}

\DeclarePairedDelimiter{\squash}{\|}{\|}
\DeclarePairedDelimiter{\tosquash}{|}{|}
\DeclareMathOperator{\powerset}{\mathcal P}
\DeclareMathOperator{\lifting}{\mathcal L}
\DeclareMathOperator{\isdefined}{is-defined}
\DeclareMathOperator{\liftvalue}{value}
\DeclareMathOperator{\const}{const}
\DeclareMathOperator{\transport}{transport}
\DeclareMathOperator{\Id}{Id}
\DeclareMathOperator{\id}{id}
\DeclareMathOperator{\ev}{ev}
\DeclareMathOperator{\inl}{inl}
\DeclareMathOperator{\inr}{inr}
\DeclareMathOperator{\isprop}{is-prop}

\newcommand{\lambdadot}[2]{\mathop{\lambda}{#1}\mathrel{.}#2}
\newcommand{\issmall}[1]{\operatorname{is}\;\!\mathcal{#1}\!\operatorname{-small}}

\newcommand{\U}{\mathcal U}
\newcommand{\V}{\mathcal V}
\newcommand{\W}{\mathcal W}
\newcommand{\T}{\mathcal T}

\newcommand{\Zero}{\mathbf{0}}
\newcommand{\Nat}{\mathbb{N}}
\newcommand{\One}{\mathbf{1}}
\newcommand{\Two}{\mathbf{2}}

\newcommand{\fst}{{\operatorname{pr_1}}}
\newcommand{\snd}{{\operatorname{pr_2}}}

\newcommand{\below}{\mathrel\sqsubseteq}
\newcommand{\aboveorder}{\mathrel\sqsupseteq}

\newcommand{\DCPOnum}[3]{{\U_{#1}}\!\operatorname{-DCPO}_{\U_{#2},\U_{#3}}}
\newcommand{\DCPO}[3]{\mathcal{#1}\!\operatorname{-DCPO}_{\mathcal{#2},\mathcal{#3}}}

\newcommand{\retract}[2]{%
\!\!\begin{tikzcd}[ampersand replacement=\&]
{#1}\ar[r,"s",shift left,hookrightarrow] %
\& {#2}\ar[l,"r",shift left,two heads]\end{tikzcd}\!\!}

\theoremstyle{plain}
\newtheorem{theorem}{Theorem}
\numberwithin{theorem}{section}
\newtheorem{lemma}[theorem]{Lemma}
\newtheorem{proposition}[theorem]{Proposition}

\theoremstyle{definition}
\newtheorem{definition}[theorem]{Definition}
\newtheorem{example}[theorem]{Example}
\theoremstyle{remark}
\newtheorem{remark}[theorem]{Remark}

\begin{document}

\title[Domain theory in UF I: Directed complete posets and Scott's \(D_\infty\)]%
{Domain theory in univalent foundations I: Directed complete posets and Scott's \(D_\infty\)}
\author[de Jong]{Tom de Jong}
\address{School of Computer Science, University of Birmingham, Birmingham, UK}
\email{\href{mailto:tom.dejong@nottingham.ac.uk}{\texttt{tom.dejong@nottingham.ac.uk}}}
\urladdr{\url{https://www.tdejong.com}}


\begin{abstract}
  We develop domain theory in constructive and predicative univalent foundations
  (also known as homotopy type theory). That we work predicatively means that we
  do not assume Voevodsky's propositional resizing axioms. Our work is
  constructive in the sense that we do not rely on excluded middle or the axiom
  of (countable) choice.
  Domain theory studies so-called directed complete posets (dcpos) and Scott
  continuous maps between them and has applications in a variety of fields, such
  as programming language semantics, higher-type computability and topology.
  A common approach to deal with size issues in a predicative foundation is to
  work with information systems, abstract bases or formal topologies rather than
  dcpos, and approximable relations rather than Scott continuous functions.
  In our type-theoretic approach, we instead accept that dcpos may be large and
  work with type universes to account for this.
  A~priori one might expect that iterative constructions of dcpos may result in a need for
  ever-increasing universes and are predicatively impossible. We show, through a
  careful tracking of type universe parameters, that such constructions can be
  carried out in a predicative setting.
  In particular, we give a predicative reconstruction of Scott's \(D_\infty\)
  model of the untyped \(\lambda\)-calculus.
  Our work is formalised in the \Agda\ proof assistant and its ability to infer
  universe levels has been invaluable for our purposes.
\end{abstract}

\maketitle

\section{Introduction}

Domain theory~\cite{AbramskyJung1994} is a well-established subject in
mathematics and theoretical computer science with applications to programming
language semantics~\cite{Scott1972,Scott1993,Plotkin1977}, higher-type
computability~\cite{LongleyNormann2015}, topology, and
more~\cite{GierzEtAl2003}.

We explore the development of domain theory from the univalent point of
view~\cite{Voevodsky2015,HoTTBook}. This means that we work with the
stratification of types as singletons, propositions, sets, 1-groupoids, etc.
Our work does not require any higher inductive types other than the
propositional truncation, and the only consequences of univalence needed here
are function extensionality and propositional extensionality.
Additionally, we work constructively and predicatively, as described below.

This paper develops the general theory of the central objects of study in domain
theory: \emph{directed complete posets (dcpos)} and the \emph{Scott continuous}
functions between them.
We show how we can predicatively construct products, exponentials, bilimits and
free pointed dcpos, culminating in the construction of Scott's famous
\(D_\infty\) which provides a mathematical model of the untyped
\(\lambda\)-calculus.

\subsection{Constructivity}
That we work constructively means that we do not assume excluded middle,
or weaker variants, such as Bishop's LPO~\cite{Bishop1967}, or the axiom of
choice (which implies excluded middle), or its weaker variants, such as the
axiom of countable choice.
An~advantage of working constructively and not relying on these additional
logical axioms is that our development is valid in every
\((\infty,1)\)-topos~\cite{Shulman2019} and not just those in which the logic is
classical.

Our commitment to constructivity has the particular consequence that we cannot
simply add a least element to a set to obtain the free pointed dcpo. Instead of
adding a single least element representing an undefined value, we must work with
a more complex type of partial elements~(\cref{sec:lifting}).
Similarly, the booleans under the natural ordering fail to be a dcpo, so we use
the type of (small) propositions, ordered by implication, instead.

\subsection{Predicativity}
Our work is predicative in the sense that we do not assume Voevodsky's
\emph{resizing} rules~\cite{Voevodsky2011,Voevodsky2015} or axioms. In
particular, powersets of small types are large.

There are several (philosophical, model-theoretic, proof-theoretic, etc.)
arguments for keeping the type theory predicative, see for
instance~\cite{Uemura2019,Swan2019a,Swan2019b},
and~\cite[Section~1.1]{deJongEscardo2023} for a brief overview, but here we only
mention one that we consider to be amongst the most interesting. Namely, the
existence of a computational interpretation of propositional impredicativity
axioms for univalent foundations is an open problem.

A common approach to deal with domain-theoretic size issues in a predicative
foundation is to work with information systems~\cite{Scott1982a,Scott1982b},
abstract bases~\cite{AbramskyJung1994} or formal
topologies~\cite{Sambin1987,Sambin2003,CoquandEtAl2003} rather than dcpos, and
approximable relations rather than \emph{Scott continuous functions}.

Instead, we work directly with dcpos and Scott continuous functions. In dealing
with size issues, we draw inspiration from category theory and make crucial use
of type universes and type equivalences to capture \emph{smallness}.
For example, in our development of the Scott model of
PCF~\cite{deJong2021a,Hart2020}, the dcpos have carriers in the second universe
\(\U_1\) and least upper bounds for directed families indexed by types in the
first universe \(\U_0\).
Moreover, up to equivalence of types, the order relation of the dcpos takes
values in the lowest universe \(\U_0\).
Seeing a poset as a category in the usual way, we can say that these dcpos are
large, but locally small, and have small filtered colimits.
The fact that the dcpos have large carriers is in fact unavoidable and
characteristic of predicative settings, as proved in~\cite{deJongEscardo2023}.

Because the dcpos have large carriers it is a priori not clear that complex
constructions of dcpos, involving countably infinite iterations of exponentials
for example, do not result in a need for ever-increasing universes and are
predicatively possible. We show that they are possible through a careful
tracking of type universe parameters, and this is also illustrated by the
construction of Scott's \(D_\infty\).

Since keeping track of these universes is prone to mistakes, we have formalised
our work in \Agda~(see~\cref{sec:formalisation}); its ability to infer and keep
track of universe levels has been invaluable.

\subsection{Related work}

In short, the distinguishing features of our work are: (i) the adoption of
homotopy type theory as a foundation, (ii) a commitment to predicatively and
constructively valid reasoning, (iii) the use of type universes to avoid size
issues concerning large posets.

The standard works on domain theory, e.g.~\cite{AbramskyJung1994,GierzEtAl2003},
are based on traditional impredicative set theory with classical logic.
A constructive, topos valid, and hence impredicative, treatment of some domain
theory can be found in~\cite[Chapter~III]{Taylor1999}.

Domain theory has been studied predicatively in the setting of formal topology
\cite{Sambin1987,Sambin2003,CoquandEtAl2003} in
\cite{MaiettiValentini2004,Negri2002,SambinValentiniVirgili1996} and the more
recent categorical papers~\cite{Kawai2017,Kawai2021}. In this predicative
setting, one avoids size issues by working with information
systems~\cite{Scott1982a,Scott1982b}, abstract bases~\cite{AbramskyJung1994} or
formal topologies, rather than dcpos, and approximable relations rather than
Scott continuous functions.
Hedberg~\cite{Hedberg1996} presented some of these ideas in Martin-L\"of Type
Theory and formalised them in the proof assistant \ALF~\cite{Magnusson1995}, a
precursor to \Agda. A~modern formalisation in \Agda\ based on Hedberg's work was
recently carried out in Lidell's master thesis~\cite{Lidell2020}.

Our development differs from the above line of work in that it studies
dcpos directly and uses type universes to account for the fact that
dcpos may be large.\index{dcpo}
An advantage of this approach is that we can work with (Scott continuous)
functions rather than the arguably more involved (approximable) relations.

Another approach to formalising domain theory in type theory can be found
in~\cite{BentonKennedyVarming2009,Dockins2014}. Both formalisations study
\(\omega\)-chain complete preorders, work with setoids, and make use of \Coq's
impredicative sort~\texttt{Prop}.
A setoid is a type equipped with an equivalence relation that must be respected
by all functions. The particular equivalence relation given by equality is
automatically respected of course, but for general equivalence relations this
must be proved explicitly.
The aforementioned formalisations work with preorders, rather than posets,
because they are setoids where two elements \(x\) and \(y\) are related if
\(x \leq y\) and \(y \leq x\).
Our~development avoids the use of setoids thanks to the adoption of the
univalent point of view. Moreover, we work predicatively and we work with the
more general directed families rather than \(\omega\)-chains, as we intend the
theory to also be applicable to topology and algebra~\cite{GierzEtAl2003}.

There are also constructive accounts of domain theory aimed at program
extraction~\cite{BauerKavkler2009,PattinsonMohammadian2021}.
Both these works study \(\omega\)-chain complete posets (\(\omega\)-cpos) and
define notions of \(\omega\)-continuity for them.
The former~\cite{BauerKavkler2009} is notably predicative, but makes use of
additional logical axioms: countable choice, dependent choice and Markov's
Principle, which are validated by a realisability interpretation.
The latter~\cite{PattinsonMohammadian2021} uses constructive logic to extract
witnesses but employs classical logic in the proofs of correctness by phrasing
them in the double negation fragment of constructive logic.
By~contrast, we study (continuous) dcpos rather than (\(\omega\)-continuous)
\(\omega\)-cpos and is fully constructive without relying on additional
principles such as countable choice or Markov's Principle.

Yet another approach is the field of \emph{synthetic domain
  theory}~\cite{Rosolini1986,Rosolini1987,Hyland1991,Reus1999,ReusStreicher1999}.
Although the work in this area is constructive, it is still impredicative, as it
is based on topos logic; but more importantly it has a focus different from that
of regular domain theory. The aim is to isolate a few basic axioms and find
models in (realisability) toposes where every object is a domain and every
morphism is continuous. These models often validate additional axioms, such as
Markov's Principle and countable choice, and moreover (necessarily) falsify
excluded middle. We have a different goal, namely to develop regular domain
theory constructively and predicatively, but in a foundation compatible with
excluded middle and choice, while not relying on them or on Markov's Principle
or countable choice.

\subsection{Relation to our other work}

This paper and its follow-up~\cite{deJongEscardoCompanion} (referred to as
Part~II) present a revised and expanded treatment of the results in our
conference paper~\cite{deJongEscardo2021a}.
In~\cite{deJongEscardo2021a} (and also \cite{deJong2021a}) the definition of a
poset included the requirement that the carrier is a set, because we only
realised later that this was redundant~(\cref{posets-are-sets}).
Products of dcpos were not discussed in these works and are included in this
paper, after having been in \Agda\ by Brendan Hart~\cite{Hart2020} for a final
year MSci project supervised by Mart\'in Escard\'o and myself.

Besides the construction of Scott's \(D_\infty\) in constructive and predicative
univalent foundations, this paper lays the foundations for
Part~II~\cite{deJongEscardoCompanion} which presents the predicative theory of
continuous and algebraic domains.
These results are also included in the author's PhD thesis~\cite{deJongThesis}.

\subsection{Formalisation}\label{sec:formalisation}
All of our results are formalised in \Agda, building on Escard\'o's \TypeTopology\
development~\cite{TypeTopology}.
Hart's previously cited work~\cite{Hart2020} was also ported to the current
\TypeTopology\ development by Escard\'o~\cite{TypeTopologyHart}.
The reference~\cite{TypeTopologyPartI} precisely links each numbered environment
(including definitions, examples and remarks) in this paper to its
implementation.
The HTML rendering has clickable links and so is particularly suitable
for exploring the development:
\url{https://www.cs.bham.ac.uk/~mhe/TypeTopology/DomainTheory.Part-I.html}.

\subsection{Organisation}
The paper is organised as follows:
\begin{description}
\item[\normalfont\cref{sec:foundations}] A brief introduction to univalent
  foundations with a particular focus on type universes and the propositional
  truncation, as well as a discussion of impredicativity in the form of
  Voevodsky's propositional resizing axioms.
\item[\normalfont\cref{sec:dcpos}] Directed complete posets (dcpos), constructively and predicatively.
\item[\normalfont\cref{sec:Scott-continuous-maps}] Scott continuous maps:
  morphism between dcpos.
\item[\normalfont\cref{sec:lifting}] The constructive lifting of a set and of a
  dcpo.
\item[\normalfont\cref{sec:products-and-exponentials}] Products and exponentials
  of dcpos.
\item[\normalfont\cref{sec:bilimits}] Bilimits of dcpos.
\item[\normalfont\cref{sec:Scott-D-infty}] A predicative reconstruction of
  Scott's \(D_\infty\) which models the untyped \(\lambda\)-calculus.
\end{description}

\section{Foundations}\label{sec:foundations}
We work within intensional Martin-L\"of Type Theory and we include \(+\)~(binary
sum), \(\Pi\)~(dependent product), \(\Sigma\)~(dependent sum), \(\Id\)
(identity type), and inductive types, including~\(\Zero\)~(empty type),
\(\One\)~(type with exactly one element \(\star : \One\)) and \(\Nat\)~(natural
numbers).
In general we adopt the same conventions of~\cite{HoTTBook}.  In particular, we
simply write \(x=y\) for the identity type \(\Id_{X}(x,y)\) and use \(\equiv\)
for the judgemental equality, and for dependent functions
\(f,g : \Pi_{x : X}A(x)\), we write \(f \sim g\) for the pointwise equality
\(\Pi_{x : X} f(x) = g(x)\).
\subsection{Universes}\label{sec:universes}
We assume a universe \(\U_0\) and two operations: for every universe \(\U\), a
successor universe \(\U^+\) with \(\U : \U^+\), and for every two universes
\(\U\) and \(\V\) another universe \(\U \sqcup \V\) such that for any
universe~\(\U\), we have \(\U_0 \sqcup \U \equiv \U\) and
\(\U \sqcup \U^+ \equiv \U^+\). Moreover, \((-)\sqcup(-)\) is idempotent,
commutative, associative, and \((-)^+\) distributes over \((-)\sqcup(-)\). We
write \(\U_1 \colonequiv \U_0^+\), \(\U_2 \colonequiv \U_1^+, \dots\) and so on.
If \(X : \U\) and \(Y : \V\), then \({X + Y} : \U \sqcup \V\) and if \(X : \U\)
and \(Y : X \to \V\), then the types \(\Sigma_{x : X} Y(x)\) and
\(\Pi_{x : X} Y(x)\) live in the universe \(\U \sqcup \V\); finally,
if~\(X : \U\) and \(x,y : X\), then \(\Id_{X}(x,y) : \U\). The type of natural
numbers \(\Nat\) is assumed to be in \(\U_0\) and we postulate that we have
copies \(\Zero_{\U}\) and \(\One_{\U}\) in every universe \(\U\).
This has the useful consequence that while we do not assume cumulativity of
universes, embeddings that lift types to higher universes are definable. For
example, the map \((-) \times \One_{\V}\) takes a type in any universe \(\U\) to
an equivalent type in the higher universe \(\U \sqcup \V\).
All our examples go through with just
two universes \(\U_0\) and \(\U_1\), but the theory is more easily developed in
a general setting.

\subsection{The univalent point of view}
Within this type theory, we adopt the univalent point of view~\cite{HoTTBook}.
A type \(X\) is a \emph{proposition} (or \emph{truth value} or
\emph{subsingleton}) if it has at most one element, i.e.\ we have an element of the type
\(\isprop(X) \colonequiv \prod_{x,y : X} x = y\).
A major difference between univalent foundations and other foundational systems
is that we \emph{prove} that types are propositions or properties. For~instance,
we can show (using function extensionality) that the axioms of directed complete
poset form a proposition.
A type \(X\) is a \emph{set} if any two elements can be identified in at most
one way, i.e.\ we have an element of the type
\(\prod_{x,y : X} \isprop(x = y)\).

\subsection{Extensionality axioms}
The univalence axiom~\cite{HoTTBook} is not needed for our development, although
we do pause to point out its consequences in two places, namely in
\cref{sec:impredicativity} and \cref{lifting-sip}.

We assume function extensionality and propositional extensionality, often
tacitly:
\begin{enumerate}[(i)]
\item \emph{Propositional extensionality}: if \(P\) and \(Q\) are two
  propositions, then we postulate that \(P = Q\) holds exactly when we have both
  \(P \to Q\) and \(Q \to P\).
\item \emph{Function extensionality}: if \(f,g : \prod_{x : X}A(x)\) are two
  (dependent) functions, then we postulate that \(f = g\) holds exactly when
  \(f \sim g\).
\end{enumerate}
Function extensionality has the important consequence that the propositions form
an exponential ideal, i.e.\ if \(X\) is a type and \(Y : X \to \U\) is such that
every \(Y(x)\) is a proposition, then so is
\(\Pi_{x : X}Y(x)\)~\cite[Example~3.6.2]{HoTTBook}. In light of this, universal
quantification is given by \(\Pi\)-types in our type~theory.

\subsection{The propositional truncation}
In Martin-L\"of Type Theory, an element of
\(\prod_{x : X}\sum_{y : Y}\phi(x,y)\), by definition, gives us a function
\(f : X \to Y\) such that \(\prod_{x : X}\phi(x,f(x))\). In some cases, we wish
to express the weaker ``for every \(x : X\), there exists some \(y : Y\) such
that \(\phi(x,y)\)'' without necessarily having an assignment of \(x\)'s to
\(y\)'s. A good example of this is when we define directed families later (see
\cref{def:directed-family}). This is achieved through the propositional truncation.

Given a type \(X : \U\), we postulate that we have a proposition
\(\squash*{X} : \U\) with a function \({\tosquash{-} : X \to \squash*{X}}\) such
that for every proposition \(P : \V\) in any universe \(\V\), every function
\(f : X \to P\) factors (necessarily uniquely, by function extensionality)
through \(\tosquash{-}\).
Diagrammatically,
\begin{equation*}
  \begin{tikzcd}
    X \ar[dr, "\tosquash*{-}"'] \ar[rr, "f"] & & P \\
    & \squash*{X} \ar[ur, dashed]
  \end{tikzcd}
\end{equation*}

Notice that the induction and recursion principles automatically hold up to an
identification: writing \(\bar f\) for the dashed map above, we have an
identification \(\bar{f}(\tosquash{x}) = f(x)\) for every \(x : X\) because
\(P\) is assumed to be a proposition.
This is sufficient for our purposes and we do not require these equalities
to hold judgementally.

Existential quantification \(\exists_{x : X}Y(x)\) is given by
\(\squash*{\Sigma_{x : X}Y(x)}\). One should note that if we have
\(\exists_{x : X}Y(x)\) and we are trying to prove some proposition \(P\), then
we may assume that we have \(x : X\) and \(y : Y(x)\) when constructing our
element of \(P\). Similarly, we can define disjunction as
\(P \lor Q \colonequiv \squash*{P + Q}\).

We assume throughout that every universe is closed under propositional
truncations, meaning that if \(X : \U\) then \(\squash{X} : \U\) as well.
We also stress that the propositional truncation is the only higher inductive
type used in our work.

\subsection{Size and impredicativity}\label{sec:impredicativity}
We introduce the notion of smallness and use it to define propositional resizing
axioms, which we take to be the definition of impredicativity in univalent
foundations.

\begin{definition}[Smallness]
    A type \(X\) in any universe is said to be \emph{\(\U\)-small} if it is
    equivalent to a type in the universe \(\U\). That is,
    \({X \issmall{\U}} \colonequiv \Sigma_{Y : \U} \pa*{Y \simeq X}\).
\end{definition}

Here, the symbol \(\simeq\) refers to Voevodsky's notion of equivalence
\cite{HoTTBook}. Notice that the type \((X \issmall{\U})\) is a proposition if
and only if the univalence axiom holds, see~\cite[Sections~3.14 and
3.36.3]{Escardo2019}.

\begin{definition}[Type of propositions \(\Omega_{\U}\)]
  The type of propositions in a universe \(\U\) is
  \(\Omega_{\U} \colonequiv \sum_{P : \U} \isprop(P) : \U^+\).
\end{definition}

Observe that \(\Omega_{\U}\) itself lives in the successor universe
\(\U^+\). We often think of the types in some fixed universe \(\U\) as
\emph{small} and accordingly we say that \(\Omega_{\U}\) is
\emph{large}.
Similarly, the powerset of a type \(X : \U\) is large.  Given our
predicative setup, we must pay attention to universes when considering
powersets:

\begin{definition}[\(\V \)-powerset \(\powerset_{\V }(X)\), \(\V \)-subsets]
  Let \(\V \) be a universe and \(X : \U \) type. We~define the
  \emph{\(\V \)-powerset} \(\powerset_{\V }(X)\) as
  \(X \to \Omega_{\V} : \V^+\sqcup \U \). Its elements are called
  \emph{\(\V \)-subsets} of \(X\).
\end{definition}
\begin{definition}[\(\in,\subseteq\)]
  Let \(x\) be an element of a type \(X\) and let \(A\) be an element of the
  powerset \(\powerset_{\V }(X)\). We write \(x \in A\) for the type
  \(\fst\pa*{A(x)}\).  The first projection \(\fst\) is needed because \(A(x)\),
  being of type \(\Omega_\V\), is a pair. Given two \(\V \)-subsets
  \(A\)~and~\(B\) of \(X\), we write \(A \subseteq B\) for
  \(\prod_{x : X}\pa*{x \in A \to x \in B}\).
\end{definition}
Function extensionality and propositional extensionality imply that \(A=B\) if and only if
\(A \subseteq B\) and \(B \subseteq A\).

One could ask for a \emph{resizing axiom} asserting that \(\Omega_{\U}\) has
size \(\U\), which we call \emph{the propositional impredicativity of \(\U\)}. A
closely related axiom is \emph{propositional resizing}, which asserts that every
proposition \(P : \U^+\) has size \(\U\). Without the addition of such resizing
axioms, the type theory is said to be \emph{predicative}.  As an example of the
use of impredicativity in mathematics, we mention that the powerset has unions
of arbitrary subsets if and only if propositional resizing
holds~\cite[Section~3.36.6]{Escardo2019}.

We note that the resizing axioms are actually theorems when classical logic
is assumed. This is because if \(P \lor \lnot P\) holds for every proposition in
\(P : \U\), then the only propositions (up to equivalence) are \(\Zero_{\U}\)
and \(\One_{\U}\), which have equivalent copies in \(\U_0\), and
\(\Omega_{\U}\) is equivalent to a type \(\Two_{\U} : \U\) with exactly two
elements.

\section{Directed complete posets}\label{sec:dcpos}

We offer the following overture in preparation of our development, especially if
the reader is familiar with domain theory in a classical, set-theoretic setting.

The basic object of study in domain theory is that of a \emph{directed complete
  poset} (dcpo).
In (impredicative) set-theoretic foundations, a dcpo can be defined to be a
poset that has least upper bounds of all directed subsets.
A naive translation of this to our foundation would be to proceed as
follows. Define a poset in a universe \(\U\) to be a type \(P:\U\) with a
reflexive, transitive and antisymmetric relation
\(-\below- : P \times P \to \U\).
Since we wish to consider posets and not categories we require that the values
\(p \below q\) of the order relation are \emph{subsingletons}.
Then we could say that the poset \((P,\below)\) is \emph{directed complete} if
every directed family \(I \to P\) with indexing type \(I : \U\) has a least
upper bound (supremum). The problem with this definition is that there are no
interesting examples in our constructive and predicative setting.
For instance, assume that the poset~$\Two$ with two elements \(0\below 1\) is
directed complete, and consider a proposition~\(A:\U\) and the directed family
\(A + \One \to \Two\) that maps the left component to~\(0\) and the right
component to~\(1\). By case analysis on its hypothetical
supremum~(\cref{def:supremum}), we conclude that the negation of \(A\) is
decidable. This amounts to weak excluded middle~(which is equivalent to De
Morgan's Law) and is constructively unacceptable.

To try to get an example, we may move to the poset \(\Omega_{\U_0}\) of
propositions in the universe \(\U_0\), ordered by implication. This poset does
have all suprema of families \(I \to \Omega_{\U_0}\) indexed by types \(I\) in
the \emph{first universe} \(\U_0\), given by existential quantification. But if
we consider a directed family \(I \to \Omega_{\U_0}\) with \(I\) in the
\emph{same universe} as \(\Omega_{\U_0}\) lives, namely the \emph{second
  universe} \(\U_1\), existential quantification gives a proposition in the
\emph{second universe} \(\U_1\) and so doesn't give its supremum. In this
example, we get a poset such that
\begin{enumerate}[(i)]
\item the carrier lives in the universe \(\U_1\),
\item the order has truth values in the universe \(\U_0\), and
\item suprema of directed families indexed by types in \(\U_0\) exist.
\end{enumerate}

Regarding a poset as a category in the usual way, we have a large, but locally
small, category with small filtered colimits (directed suprema). This is typical
of all the concrete examples that we consider, such as the dcpos in the
Scott model of PCF~\cite{deJong2021a} and Scott's \(D_\infty\)
model of the untyped \(\lambda\)-calculus (\cref{sec:Scott-D-infty}).
We may say that the predicativity restriction increases
the universe usage by one.  However, for the sake of generality, we formulate
our definition of dcpo with the following universe conventions:
\begin{enumerate}[(i)]
\item the carrier lives in a universe \(\U\),
\item the order has truth values in a universe \(\T\), and
\item suprema of directed families indexed by types in a universe \(\V\) exist.
\end{enumerate}
So our notion of dcpo has three universe parameters \(\U,\V\) and \(\T\). We
will say that the dcpo is \emph{locally small} when \(\T\) is not necessarily
the same as \(\V\), but the order has \(\V\)-small truth values. Most of the
time we mention \(\V\) explicitly and leave \(\U\) and \(\T\) to be understood
from the context.

We now define directed complete poset in constructive and predicative univalent
foundations. We carefully explain our use of the propositional truncation in our
definitions and, as mentioned above, the type universes involved.

\begin{definition}[Preorder, reflexivity, transitivity]
  A \emph{preorder} \((P,\sqsubseteq)\) is a type \(P : \U \) together with a
  proposition-valued binary relation \({\sqsubseteq} : {P \to P \to \Omega_\T}\)
  satisfying
  \begin{enumerate}[(i)]
  \item \emph{reflexivity}: for every \(p : P\), we have \(p \below p\), and%
  \item \emph{transitivity}: for every \(p,q,r : P\), if \(p \below q\) and
    \(q \below r\), then \(p \below r\).%
    \qedhere
  \end{enumerate}
\end{definition}

\begin{definition}[Poset, antisymmetry]
  A \emph{poset} is a preorder \((P,\below)\) that is \emph{antisymmetric}: if
  \(p \below q\) and \(q \below p\), then \(p = q\) for every \(p,q : P\).
\end{definition}

\begin{lemma}\label{posets-are-sets}
  If \((P,\below)\) is a poset, then \(P\) is a set.
\end{lemma}
\begin{proof}
  For every \(p,q : P\), the composite
  \[
    \pa{p = q} \xrightarrow{\text{by reflexivity}}
    {\pa{p \below q} \times \pa{q \below p}} \xrightarrow{\text{by antisymmetry}}
    \pa{p = q}
  \]
  is constant since \({\pa{p \below q} \times \pa{q \below p}}\) is a
  proposition. By \cite[Lemma~3.11]{KrausEtAl2017} it therefore follows that
  \(P\) must be a set.
\end{proof}

From now on, we will simply write ``let \(P\) be a poset'' leaving the partial
order \(\below\) implicit. We will often use the symbol \({\below}\) for partial
orders on different carriers when it is clear from the context which one it
refers to.

\begin{definition}[(Semi)directed family]\label{def:directed-family}
  A family \(\alpha : I \to P\) of elements of a poset \(P\) is
  \emph{semidirected} if whenever we have \(i,j : I\), there exists some
  \(k : I\) such that \(\alpha_i \below \alpha_k\) and
  \(\alpha_j \below \alpha_k\).
  We frequently use the shorthand \({\alpha_i,\alpha_j} \below \alpha_k\) to
  denote the latter requirement.
  Such a family is \emph{directed} if it is semidirected and its domain \(I\) is
  inhabited.%
\end{definition}

The name ``semidirected'' matches Taylor's terminology~\cite[Definition~3.4.1]{Taylor1999}.

\begin{remark}
  Note our use of the propositional truncation in defining when a family is
  \emph{directed}. To make this explicit, we write out the definition in
  type-theoretic syntax: a family \(\alpha : I \to P\) is directed if
  \begin{enumerate}[(i)]
  \item\label{dir-inh} we have an element of \(\squash{I}\), and
  \item\label{dir-semidir}
    \(\Pi_{i,j : I} \squash*{\Sigma_{k : I}\pa*{\alpha_i \below \alpha_k} \times
      \pa*{\alpha_j \below \alpha_k}}\).
  \end{enumerate}
  The use of the propositional truncation ensures that the types \eqref{dir-inh}
  and \eqref{dir-semidir} are propositions and hence that being (semi)directed
  is a property of a family.
  The type \eqref{dir-semidir} without truncation would instead express an
  assignment of a chosen \(k : I\) for every \(i,j : I\) instead.
\end{remark}

\begin{definition}[(Least) upper bound, supremum]\label{def:supremum}
  An element \(x\) of a poset \(P\) is an \emph{upper bound} of a family
  \(\alpha : I \to P\) if \(\alpha_i \below x\) for every \(i : I\).
  It is a \emph{least upper bound} of~\(\alpha\) if it is an upper bound, and
  whenever \(y : P\) is an upper bound of \(\alpha\), then \(x \below y\).
  By antisymmetry, a least upper bound is unique if it exists, so in this case
  we will speak of \emph{the} least upper bound of \(\alpha\), or sometimes the
  \emph{supremum} of \(\alpha\).
\end{definition}

\begin{definition}[\(\V \)-directed complete poset, \(\V\)-dcpo, %
  \(\bigsqcup \alpha\), \(\bigsqcup_{i : I}\alpha_i\)]
  For a universe~\(\V\), a \emph{\(\V \)-directed complete poset} (or
  \emph{\(\V \)-dcpo}, for short) is a poset \(D\) such that every directed
  family \(\alpha : I \to D\) with \(I : \V \) has a supremum in \(D\) that we
  denote by \(\bigsqcup \alpha\) or \(\bigsqcup_{i : I} \alpha_i\).
\end{definition}

\begin{remark}\label{directed-completeness-is-prop}
  Explicitly, we ask for an element of the type
  \[
    \Pi_{I : \V}\Pi_{\alpha : I \to D}\pa*{\operatorname{is-directed} \alpha \to
      \Sigma_{x : D}\pa*{x \mathrel{\operatorname{is-sup-of}} \alpha}},
  \]
  where \(\pa*{x \mathrel{\operatorname{is-sup-of}} \alpha}\) is the type expressing
  that \(x\) is the supremum of \(\alpha\).
  Even though we used \(\Sigma\) and not \(\exists\) in this expression, this
  type is still a proposition: By \cite[Example~3.6.2]{HoTTBook}, it suffices to
  prove that the type
  \(\Sigma_{x : D}(x \mathrel{\operatorname{is-sup-of}} \alpha)\) is a
  proposition. So suppose that we have \(x,y : D\) with
  \(p : x \mathrel{\operatorname{is-sup-of}} \alpha\) and
  \(q : y \mathrel{\operatorname{is-sup-of}} \alpha\). Being the supremum of a
  family is a property because the partial order is proposition-valued. Hence,
  by \cite[Lemma~3.5.1]{HoTTBook}, to prove that \((x,p) = (y,q)\), it suffices
  to prove that \(x = y\). But this follows from antisymmetry and the fact that
  \(x\) and \(y\) are both suprema of \(\alpha\).
\end{remark}

We will sometimes leave the universe \(\V \) implicit, and simply speak of a
dcpo. On other occasions, we need to carefully keep track of universe levels. To
this end, we make the following definition.
\begin{definition}[\(\DCPO{V}{U}{T}\)]
  Let \(\V\), \(\U\) and \(\T \) be universes. We write \(\DCPO{V}{U}{T}\) for
  the type of \(\V \)-dcpos with carrier in \(\U \) and order taking values in
  \(\T \).
\end{definition}

\begin{remark}\label{universe-levels-of-lifting-and-exponentials}
  In particular, it is very important to keep track of the universe parameters
  of the lifting~(\cref{sec:lifting}) and of
  exponentials~(\cref{sec:products-and-exponentials}) in order to ensure that it
  is possible to construct Scott's \(D_\infty\) (and the Scott model of
  PCF~\cite{deJong2021a}) in our predicative setting, as we do
  in~\cref{sec:Scott-D-infty}.
\end{remark}

In many examples and applications, we require our dcpos to have a least element.

\begin{definition}[Pointed dcpo]
  A dcpo \(D\) is \emph{pointed} if it has a least element which we will denote
  by \(\bot_{D}\), or simply \(\bot\).
\end{definition}

\begin{definition}[Local smallness]\label{def:local-smallness}
  A \(\V\)-dcpo \(D\) is \emph{locally small} if \(x \below y\) is \(\V\)-small
  for every \(x,y : D\).
\end{definition}

\begin{lemma}\label{local-smallness-alt}
  A \(\V\)-dcpo \(D\) is locally small if and only if we have
  \({\below_{\V}} : D \to D \to \V\) such that \(x \below y\) holds
  precisely when \(x \below_{\V} y\) does.
\end{lemma}
\begin{proof}
  The \(\V\)-dcpo \(D\) is locally small exactly when we have an element of
  \[
    \Pi_{x,y : D}\Sigma_{T : \V}\pa*{T \simeq {x \below y}}.
  \]
  But this type is equivalent to
  \[
    \Sigma_{R : {D \to D \to \V}}\Pi_{x,y : D}\pa*{R(x,y) \simeq {x \below y}}
  \]
  by distributivity of \(\Pi\) over \(\Sigma\)~\cite[Theorem~2.5.17]{HoTTBook}.
\end{proof}

Nearly all examples of \(\V\)-dcpos in this thesis will be locally small. We now
introduce two fundamental examples of dcpos: the type of subsingletons and
powersets.

\begin{example}[The type of subsingletons as a pointed dcpo]\label{Omega-as-pointed-dcpo}
  For any type universe~\(\V\), the type \(\Omega_{\V}\) of subsingletons in
  \(\V\) is a poset if we order the propositions by implication.
  Note that antisymmetry holds precisely because of propositional
  extensionality.
  Moreover, \(\Omega_{\V}\) has a least element, namely \(\Zero_{\V}\), the
  empty type in~\(\V\).
  We also claim that \(\Omega_{\V}\) has suprema for all (not necessarily
  directed) families \(\alpha : I \to \Omega_{\V}\) with \(I : \V\).
  Given such a family \(\alpha\), its least upper bound is given by
  \(\exists_{i : I}\,\alpha_i\). It is clear that this is indeed an upper bound
  for \(\alpha\). And if \(P\) is a subsingleton such that \(\alpha_i \below P\)
  for every \(i : I\), then to show that
  \(\pa*{\exists_{i : I}\,\alpha_i} \to P\) it suffices to construct to
  construct a map \(\pa*{\Sigma_{i : I}\,\alpha_i} \to P\) as \(P\) is a
  subsingleton. But this is easy because we assumed that \(\alpha_i \below P\)
  for every \(i : I\).
  Finally, paying attention to the universe levels we observe that
  \(\Omega_{\V} : \DCPO{V}{V^+}{V}\).
\end{example}

\begin{example}[Powersets as pointed dcpos]\label{powersets-as-pointed-dcpos}
  Recalling our treatment of subset and powersets from
  \cref{sec:impredicativity}, we show that powersets give examples of
  pointed dcpos.
  Specifically, for every type \(X : \U\) and every type universe \(\V\), the
  subset inclusion \(\subseteq\) makes \(\powerset_{\V}(X)\) into a poset, where
  antisymmetry holds by function extensionality and propositional
  extensionality.
  Moreover, \(\powerset_{\V}(X)\) has a least element of course: the empty set
  \(\emptyset\).
  We also claim that \(\powerset_{\V}(X)\) has suprema for all (not necessarily
  directed) families \(\alpha : I \to \powerset_{\V}(X)\) with \(I : \V\).
  Given such a family \(\alpha\), its least upper bound is given by
  \(\bigcup \alpha \colonequiv \lambdadot{x}{\exists_{i : I}\,x\in\alpha_i}\),
  the set-theoretic union, which is well-defined as
  \(\pa*{\exists_{i : I}\,x\in\alpha_i} : \V\).
  It is clear that this is indeed an upper bound for \(\alpha\). And if \(A\) is
  a \(\V\)-subset of \(X\) such that \(\alpha_i \subseteq A\) for every
  \(i : I\), then to show that \(\bigcup \alpha \subseteq A\) it suffices to
  construct for every \(x : X\), a map
  \(\pa*{\Sigma_{i : I}\,\alpha_i} \to \pa*{x\in A}\) as \(x \in A\) is a
  subsingleton. But this is easy because we assumed that
  \(\alpha_i \subseteq A\) for every \(i : I\).
  Finally, paying attention to the universe levels we observe that
  \(\powerset_{\V}(X) : \DCPO{V}{V^+ \sqcup \U}{V \sqcup \U}\).
  In the case that \(X : \U \equiv \V\), we obtain the simpler, locally small
  \(\powerset_{\V}(X) : \DCPO{V}{V^+}{V}\).
\end{example}

Of course, \(\Omega_{\V}\) is easily seen to be equivalent to
\(\powerset_{\V}(\One_\V)\), so \cref{powersets-as-pointed-dcpos} subsumes
\cref{Omega-as-pointed-dcpo}, but it is instructive to understand
\cref{Omega-as-pointed-dcpo} first.

\begin{proposition}[\(\omega\)-completeness]\label{dcpo-has-sups-of-chains}
  Every \(\V\)-dcpo \(D\) is \emph{\(\omega\)-complete}, viz.\ if we
  have elements \(x_0 \below x_1 \below x_2 \below \dots\) of \(D\), then the
  supremum of \(\pa*{x_n}_{n : \Nat}\) exists in \(D\).
\end{proposition}
\begin{proof}
  Recalling from~\cref{sec:universes} that we can lift types to higher universes
  and using the fact that \(\Nat\) is in the base universe \(\U_0\), we obtain a
  type \(\operatorname{lift}_{\U_0,\V}(\Nat) \) in the universe \(\V\) that is
  equivalent to \(\Nat\).
  Now
  \(\operatorname{lift}_{\U_0,\V}(\Nat) \simeq \Nat \xrightarrow{x_{(-)}} D\) is
  a directed family as \(x_n \below x_{n+1}\) for every natural number \(n\),
  and it is indexed by a type in \(\V\).
  Hence, it has a least upper bound in \(D\) which is the supremum of
  \(\pa*{x_n}_{n : \Nat}\).
\end{proof}

\section{Scott continuous maps}\label{sec:Scott-continuous-maps}

We discuss an appropriate notion of morphism between \(\V\)-dcpos, namely one
that requires preservation of directed suprema and the order
(\cref{continuous-implies-monotone}).

\begin{definition}[Scott continuity]
  A function \(f : D \to E\) between two \(\V\)-dcpos is \emph{(Scott)
    continuous} if it preserves directed suprema, i.e.\ if \(I : \V \) and
  \(\alpha : I \to D\) is directed, then \(f\pa*{\bigsqcup \alpha}\) is the
  supremum in \(E\) of the family \(f \circ \alpha\).
\end{definition}

\begin{remark}
  When we speak of a Scott continuous function between \(D\) and \(E\), then we
  will always assume that \(D\) and \(E\) are both \(\V\)-dcpos for some
  arbitrary but fixed type universe \(\V\).
\end{remark}

\begin{lemma}
  Being Scott continuous is a property. In particular, two Scott continuous maps
  are equal if and only if they are equal as functions.
\end{lemma}
\begin{proof}
  By \cite[Example~3.6.2]{HoTTBook} and the fact that being the supremum of a family is a
  property, cf.\ \cref{directed-completeness-is-prop}.
\end{proof}

\begin{lemma}\label{continuous-implies-monotone}
  If \(f : D \to E\) is Scott continuous, then it is \emph{monotone}, i.e.\
  \(x \below_{D} y\) implies \(f(x) \below_{E} f(y)\).
\end{lemma}
\begin{proof}
  Given \(x,y : D\) with \(x \below y\), consider the directed family
  \(\Two_{\V} \xrightarrow{\alpha} D\) defined by \(\alpha(0) \colonequiv x\)
  and \(\alpha(1) \colonequiv y\). Its supremum is \(y\) and \(f\) must preserve
  it. Hence, \(f(y)\) is an upper bound of \(f(\alpha(0)) \equiv f(x)\), so
  \(f(x) \below f(y)\), as we wished to show.
\end{proof}

\begin{lemma}\label{image-is-directed}
  If \(f : D \to E\) is monotone and \(\alpha : I \to D\) is directed, then so
  is \(f \circ \alpha\).
\end{lemma}
\begin{proof}
  Since \(\alpha\) is directed, \(I\) is inhabited, so it remains to prove that
  \(f \circ \alpha\) is semidirected. If we have \(i,j : I\), then by
  directedness of \(\alpha\), there exists \(k : I\) such that
  \({\alpha_i,\alpha_j}\below\alpha_k\). By monotonicity,
  we obtain \({f(\alpha_i) , f(\alpha_j)} \below f(\alpha_k)\) as desired.
\end{proof}

\begin{lemma}\label{continuity-criterion}
  A monotone map \(f : D \to E\) between \(\V\)-dcpos is Scott continuous if and
  only if \(f\pa*{\bigsqcup \alpha} \below \bigsqcup {f \circ \alpha}\).
\end{lemma}

Note that we are justified in writing \(\bigsqcup {f \circ \alpha}\) because
\cref{image-is-directed} tells us that \(f \circ \alpha\) is directed by the
assumed monotonicity of \(f\).
\begin{proof}
  The left-to-right implication is immediate. For the converse, note that it
  only remains to show that
  \(f\pa*{\bigsqcup \alpha} \aboveorder \bigsqcup{f \circ \alpha}\).
  But for this it suffices that
  \(f\pa*{\alpha_i} \below f\pa*{\bigsqcup \alpha}\) for every \(i : I\), which
  holds as \(\bigsqcup \alpha\) is an upper bound of \(\alpha\) and \(f\) is
  monotone.
\end{proof}

\begin{remark}
  In constructive mathematics it is not possible to exhibit a discontinuous
  function from \(\Nat^\Nat\) to \(\Nat\), because
  sheaf~\cite[Chapter~15]{TroelstraVanDalen1988} and realizability
  models~\cite[e.g.~Proposition~3.1.6]{vanOosten2008} imply that it is
  consistent to assume that all such functions are continuous.
  This does not mean, however, that we cannot exhibit a discontinuous function
  between dcpos. In fact, the negation map \({\lnot} : \Omega \to \Omega\) is
  not monotone and hence not continuous.
  If we were to preclude such examples, then we can no longer work with the full
  type \(\Omega\) of all propositions, but instead we must restrict to a subtype
  of propositions, for example by using dominances~\cite{Rosolini1986}.
  Indeed, this approach is investigated in the context of topos theory
  in~\cite{Phao1991,Longley1995} and for computability instead of continuity in
  univalent foundations in~\cite{EscardoKnapp2017}.
\end{remark}

\begin{definition}[Strictness]
  A Scott continuous function \(f : D \to E\) between pointed dcpos is \emph{strict}
  if \(f\pa*{\bot_{D}} = \bot_{E}\).
\end{definition}

\begin{lemma}\label{pointed-dcpos-sups}
  A poset \(D\) is a pointed \(\V\)-dcpo if and only if it has suprema for all
  semidirected families indexed by types in \(\V\) that we will denote using the
  \(\bigvee\) symbol.
  In~particular, a pointed \(\V\)-dcpo has suprema of all families indexed by
  propositions in \(\V\).

  Moreover, if \(f\) is a Scott continuous and strict map between pointed
  \(\V\)-dcpos, then \(f\) preserves suprema of semidirected families.
\end{lemma}
\begin{proof}
  If \(D\) is complete with respect to semidirected families indexed by types
  in~\(\V\), then it is clearly a \(\V\)-dcpo and it is pointed because the
  supremum of the family indexed by the empty type is the least element.
  Conversely, if \(D\) is a pointed \(\V\)-dcpo and \(\alpha : I \to D\) is a
  semidirected family with \(I : \V\), then the family
  \begin{align*}
    \hat\alpha : I + \One_{\V} &\to D \\
    \inl(i) &\mapsto \alpha_i \\
    \inr(\star) &\mapsto \bot
  \end{align*}
  is directed and hence has a supremum in \(D\) which is also the least
  upper bound of~\(\alpha\).

  A pointed \(\V\)-dcpo must have suprema for all families indexed by
  propositions in~\(\V\), because any such family is semidirected.
  Finally, suppose that \(\alpha : I \to D\) is semidirected and that
  \(f : D \to E\) is Scott continuous and strict. Using the
  \(\widehat{(-)}\)-construction from above, we see that
  \begin{align*}
    f\pa*{\textstyle\bigvee \alpha}
    &\equiv f\pa*{\textstyle\bigsqcup\hat\alpha} \\
    &= \textstyle\bigsqcup f \circ \hat\alpha
    &&\text{(by Scott continuity of \(f\))} \\
    &= \textstyle\bigsqcup \widehat{f \circ \alpha}
    &&\text{(since \(f\) is strict)} \\
    &\equiv \textstyle\bigvee {f \circ \alpha},
  \end{align*}
  finishing the proof.
\end{proof}

\begin{proposition}\label{continuity-closure}%
  \hfill
  \begin{enumerate}[(i)]
  \item\label{id-is-continuous} The identity on any dcpo is Scott continuous.
  \item\label{const-is-continuous} For dcpos \(D\) and \(E\) and \(y : E\), the
    constant map \(x \mapsto y : D \to E\) is Scott continuous.
  \item\label{comp-is-continuous} If \(f : D \to E\) and \(g : E \to E'\) are
    Scott continuous, then so is \(g \circ f\).
  \end{enumerate}
  Moreover, if \(D\) is pointed, then the identity on \(D\) is strict, and if
  \(f\)~and~\(g\) are strict in \eqref{comp-is-continuous}, then so is
  \(g \circ f\).
\end{proposition}
\begin{proof}
  The proofs of~\eqref{id-is-continuous}~and~\eqref{const-is-continuous} are
  obvious. For~\eqref{comp-is-continuous}, let \(\alpha : I \to D\) be directed
  and notice that
  \(g\pa*{f\pa*{\textstyle\bigsqcup \alpha}} = g\pa*{\textstyle\bigsqcup f \circ
    \alpha} = \textstyle\bigsqcup {g \circ f \circ \alpha}\) by respectively
  continuity of \(f\)~and~\(g\). The claims about strictness are also clear.
\end{proof}

\begin{definition}[Isomorphism]
  A Scott continuous map \(f : D \to E\) is an \emph{isomorphism} if we have a
  Scott continuous inverse \(g : E \to D\).
\end{definition}

\begin{lemma}\label{isomorphism-is-strict}
  Every \(f : D \to E\) isomorphism between pointed dcpos is strict.
\end{lemma}
\begin{proof}
  Let \(y : E\) be arbitrary and notice that \(\bot_D \below g(y)\) because
  \(\bot_D\) is the least element of \(D\). By monotonicity of \(f\), we get
  \(f(\bot_D) \below f(g(y)) = y\) which shows that \(f(\bot_D)\) is the least
  element of \(E\).
\end{proof}

\begin{definition}[Scott continuous retract]\label{def:continuous-retract}%
  A dcpo \(D\) is a \emph{(Scott) continuous retract} of \(E\) if we have Scott
  continuous maps \(s : D \to E\) and \(r : E \to D\) such that \(s\)~is a
  section of \(r\). We denote this situation by \(\retract{D}{E}\).
\end{definition}

\begin{lemma}\label{locally-small-retract}
  If \(D\) is a continuous retract of \(E\) and \(E\) is locally small,
  then so is~\(D\).
\end{lemma}
\begin{proof}
  We claim that \(x \below_D y\) and \(s(x) \below_E s(y)\) are equivalent,
  which proves the lemma as \(E\) is assumed to be locally small.
  One direction of the equivalence is given by the fact that \(s\) is monotone.
  In the other direction, assume that \(s(x) \below s(y)\) and note that
  \(x = r(s(x)) \below r(s(y)) = y\), as \(r\) is monotone and \(s\) is a
  section of \(r\).
\end{proof}

\section{Lifting}\label{sec:lifting}

We now turn to constructing pointed \(\V\)-dcpos from sets.
First of all, every discretely ordered set is a \(\V\)-dcpo, where discretely
ordered means that we have \(x \below y\) exactly when \(x = y\). This is
because if \(\alpha : I \to X\) is a directed family into a discretely ordered
set~\(X\), then \(\alpha\) has to be constant (by semidirectedness), so
\(\alpha_i\) is its supremum for any \(i : I\). And since directedness includes
the condition that the domain is inhabited, it follows that \(\alpha\) must have
a supremum in \(X\).
In fact, ordering \(X\) discretely yields the free \(\V\)-dcpo on the set \(X\)
in the categorical sense.

With excluded middle, the situation for \emph{pointed} \(\V\)-dcpos is also very
straightforward. Simply order the set \(X\) discretely and add a least element,
as depicted for \(X \equiv \Nat\) in the Hasse diagram of~\cref{flat-gives-LPO}.
The point of that proposition is to show, by a reduction to the constructive
taboo LPO~\cite{Bishop1967}, that this approach is constructively unsatisfactory.
In fact, in~\cite{deJongEscardo2023} we proved a general constructive no-go
theorem showing that there is a nontrivial dcpo with decidable equality if and
only if weak excluded middle holds.

\begin{proposition}\label{flat-gives-LPO}
  If the poset \(\Nat_\bot = (\Nat + \One,{\below})\) with order depicted by the Hasse diagram
  
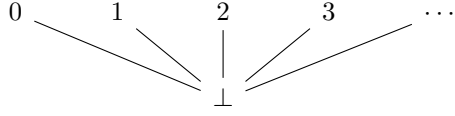
\[
    \begin{tikzcd}
      0 \ar[drr,dash] & 1 \ar[dr,dash] & 2 \ar[d,dash] & 3 \ar[dl,dash] & \cdots \ar[dll,dash] \\
      & & \bot
    \end{tikzcd}
  \]
  is \(\omega\)-complete, then LPO holds.
  In particular, by \cref{dcpo-has-sups-of-chains}, if it is
  \(\U_0\)-directed complete, then LPO holds.
\end{proposition}
\begin{proof}
  Let \(\alpha : \Nat \to \Two\) be an arbitrary binary sequence. We show that
  \(\exists_{n : \Nat}\,\alpha_n = 1\) is decidable.
  Define the family \(\beta : \Nat \to \Nat_\bot\) by
  \[
    \beta_n \colonequiv
    \begin{cases}
      \inl(k) &\text{if \(k\) is the least integer below \(n\) for which \(\alpha_k = 1\), and} \\
      \inr(\star) &\text{else}.
    \end{cases}
  \]
  Then \(\beta\) is a chain, so by assumption it has a supremum \(s\) in \(\Nat_\bot\).
  By the induction principle of coproducts, we have \(s = \inr(\star)\) or we
  have \(k : \Nat\) such that \(s = \inl(k)\).
  If the latter holds, then \(\alpha_k = 1\), so
  \(\exists_{n : \Nat}\,\alpha_n = 1\) is decidable.
  We claim that \(s = \inr(\star)\) implies that
  \(\lnot\pa*{\exists_{n : \Nat}\,\alpha_n = 1}\).
  Indeed, assume for a contradiction that \(\exists_{n : \Nat}\,\alpha_n =
  1\). Since we are proving a proposition, we may assume to have \(n : \Nat\)
  with \(\alpha_n = 1\).
  Then, \(\beta_n = \inl(k)\) for a natural number \(k \leq n\). Since \(s\) is
  the supremum of \(\beta\) we have
  \(\inl(k) = \beta_n \below s = \inr(\star)\), but \(\inr(\star)\) is the least
  element of \(\Nat_\bot\), so by antisymmetry \(\inl(k) = \inr(\star)\), which is
  impossible.
\end{proof}

Our solution to the above will be to work with the lifting monad, sometimes
known as the partial map classifier monad from topos
theory~\cite{Johnstone1977,Rosolini1986,Rosolini1987,Kock1991}, which has been
extended to constructive type theory by
\citeauthor{ReusStreicher1999}~\cite{ReusStreicher1999} and recently to
univalent foundations by
\citeauthor{EscardoKnapp2017}~\cite{EscardoKnapp2017,Knapp2018}.

\begin{remark}
  Escard\'o and Knapp's work~\cite{EscardoKnapp2017,Knapp2018} yields an
  approach to partiality in univalent foundations and aim to avoid (weak)
  countable choice, which is not provable in constructive univalent
  foundations~\cite{CoquandMannaaRuch2017,Coquand2018,Swan2019a,Swan2019b}.
  This is to be contrasted to other approaches to partiality in Martin-L\"of
  Type Theory. The first is Capretta's delay monad~\cite{Capretta2005}, which
  uses coinduction.
  Arguably, the correct notion of equality of Capretta's delay monad is that of
  weak bisimilarity where two partial elements are considered equal when they
  are both undefined, or, when one of them is, so is the other and they have
  equal values in this case.
  This prompted the authors of~\cite{ChapmanUustaluVeltri2019} to consider its
  quotient by weak bisimilarity, but they use countable choice to show that
  quotient is again a monad. Again using countable choice, they show that their
  quotient yields free pointed \(\omega\)-complete posets (\(\omega\)-cpos).%
  \index{omega-completeness@\(\omega\)-completeness}
  In~\cite{AltenkirchDanielssonKraus2017} the authors use a so-called quotient
  inductive-inductive type (QIIT) to construct the free pointed \(\omega\)-cpo,
  essentially by definition of the QIIT. It~was shown
  in~\cite{ChapmanUustaluVeltri2019} that a simpler higher inductive type
  actually suffices.
  Regardless, we stress that our approach yields free dcpos as opposed to
  \(\omega\)-cpos and does not use countable choice or higher inductive types
  other than the propositional truncation.
  For a discussion on \(\omega\)-completeness, countable choice and their
  relations to the Scott model of PCF, see~\cite[Section~8]{deJong2021a}
  and~\cite[Section~5.2.4]{deJongThesis}.
\end{remark}

\begin{definition}[Lifting, partial element, \(\lifting_{\V}(X)\); %
  {\cite[Section~2.2]{EscardoKnapp2017}}]
  We define the type of \emph{partial elements} of a type \(X : \U\) with
  respect to a universe \(\V\) as
  \[
    \lifting_{\V}(X) \colonequiv \Sigma_{P : \Omega_{\V}}(P \to X)
  \]
  and we also call it the \emph{lifting} of \(X\) with respect to \(\V\).
\end{definition}

Every (total) element of \(X\) gives rise to a partial element of \(X\) through
the following map, which will be shown to be the unit of the monad later.

\begin{definition}[\(\eta_X\)]
  The map \(\eta_X : X \to \lifting_{\V}(X)\) is defined by mapping \(x\) to the
  tuple \(\pa*{\One_{\V},\lambdadot{u}{x}}\), where we have omitted the witness
  that \(\One_{\V}\) is a subsingleton.
  We sometimes omit the subscript in \(\eta_X\).
\end{definition}

Besides these total elements, the lifting has another distinguished element that
will be the least with respect to the order with which we shall equip the
lifting.

\begin{definition}[\(\bot\)]\label{def:lifting-bot}
  For every type \(X : \U\) and universe \(\V\), we denote the element
  \(\pa*{\Zero_{\V},\varphi} : \lifting_{\V}(X)\) by \(\bot\). (Here \(\varphi\)
  is the unique map from \(\Zero_{\V}\) to \(X\).)
\end{definition}

Next we introduce appropriate names for the projections from the type of partial
elements.
\begin{definition}[\(\isdefined\) and \(\liftvalue\)]
  We write \(\isdefined : \lifting_{\V}(X) \to \Omega_{\V}\) for the first
  projection and
  \(\liftvalue : \Pi_{l : \lifting_{\V}(X)}\pa*{\isdefined(l) \to X}\) for the
  second projection.
\end{definition}

Thus, with this terminology, the element \(\star\) witnesses that \(\eta(x)\) is
defined with value \(x\) for every \(x : X\), while \(\bot\) is not defined
because \(\isdefined(\bot)\) is the empty type.

Excluded middle says exactly that such elements are the only elements of the
lifting of a type \(X\), as the following proposition shows.
Thus, the lifting generalises the classical construction of adding a new
element.

\begin{proposition}[{\cite[Section~2.2]{EscardoKnapp2017}}]%
  \label{lifting-is-plus-one-iff-em}
  The map \(X + \One \xrightarrow{[\eta,\const_{\bot}]} \lifting_{\V}(X)\) is an
  equivalence for every type \(X : \U\) if and only if excluded middle in \(\V\)
  holds.
\end{proposition}
\begin{proof}
  Excluded middle in \(\V\) is equivalent to the map
  \([\const_{\Zero},\const_{\One}] : \Two_{\V} \to \Omega_{\V}\) being an
  equivalence.
  But if that map is an equivalence, then it follows that the map
  \([\eta,\const_{\bot}]: X + \One \to \lifting_{\V}(X)\) is also an equivalence
  for every type \(X\).
  Conversely, we can take \(X \colonequiv \One_{\V}\) to see that
  \([\const_{\Zero},\const_{\One}] : \Two_{\V} \to \Omega_{\V}\) must be an
  equivalence.
\end{proof}

\begin{lemma}\label{lifting-equality}
  Two partial elements \(l,m : \lifting_{\V}(X)\) of a type \(X\) are equal if
  and only if we have \(\isdefined(l) \iff \isdefined(m)\) and the diagram
  \[
    \begin{tikzcd}
      \isdefined(m) \ar[dr] \ar[rr,"\liftvalue(m)"] & & X \\
      & \isdefined(l) \ar[ur,"\liftvalue(l)"']
    \end{tikzcd}
  \]
  commutes.
\end{lemma}
\begin{proof}
  By the characterisation of the identity type of
  \(\Sigma\)-types~\cite[Theorem~2.7.2]{HoTTBook}, we have
  \[
    (l = m) \simeq \pa*{\Sigma_{e : \isdefined(l) = \isdefined(m)}\,
      \transport^{\lambdadot{P}{P \to X}}(e,\liftvalue(l)) = \liftvalue(m)}
  \]
  By path induction on \(e\) we can prove that
  \[
    \transport^{\lambdadot{P}{P \to X}}(e,\liftvalue(l)) = \liftvalue(l) \circ \tilde{e}^{-1},
  \]
  where \(\tilde{e}\) is the equivalence \(\isdefined(l) \simeq \isdefined(m)\)
  induced by \(e\).
  Hence, using function extensionality and propositional extensionality, the
  right hand side of the equivalence given above is logically equivalent to
  \[
    \Sigma_{(e_1,e_2) :
      \isdefined(l)\leftrightarrow\isdefined(m)}\,\liftvalue(l) \circ e_2 \sim
    \liftvalue(m),
  \]
  as desired.
\end{proof}

\begin{remark}\label{lifting-sip}
  It is possible to promote the logical equivalence of \cref{lifting-equality}
  to an equivalence of types using univalence and a generalised \emph{structure
    identity principle}~\cite[Section~3.33]{Escardo2019}, as done
  in~\cite[Lemma~44]{Escardo2021} and formalised in~\cite[\mkTTurl{Lifting.IdentityViaSIP}]{TypeTopology}.
  But the above logical equivalence will suffice for our purposes.
\end{remark}

\begin{theorem}[Lifting monad, Kleisli extension, \(f^\#\);
  {\cite[Section~2.2]{EscardoKnapp2017}}]\label{lifting-is-monad}
  The lifting is a monad with unit \(\eta\). That~is, for every map
  \(f : X \to \lifting_{\V}(Y)\) we have a map
  \(f^\# : \lifting_{\V}(X) \to \lifting_{\V}(Y)\), the \emph{Kleisli extension}
  of \(f\), such that
  \begin{enumerate}[(i)]
  \item\label{eta-ext} \(\eta_X^\# \sim \id_{\lifting_{\V}(X)}\) for every type
    \(X\),
  \item\label{ext-eta} \(f^\# \circ \eta_X \sim f\) for every map
    \(f : X \to \lifting_{\V}(Y)\), and
  \item\label{comp-ext} \((g^\# \circ f)^\# \sim g^\# \circ f^\#\) for every two
    maps \(f : X \to \lifting_{\V}(Y)\) and \(g : Y \to \lifting_{\V}(Z)\).
  \end{enumerate}
\end{theorem}
\begin{proof}
  Given \(f : X \to \lifting_{\V}(Y)\), we define
  \begin{align*}
    f^\# : \lifting_{\V}(X) &\to \lifting_{\V}(Y)\\
    (P,\varphi) &\mapsto \pa*{\Sigma_{p : P}\isdefined(f(\varphi(p))),\psi},
  \end{align*}
  where \(\psi(p,q) \colonequiv \liftvalue\pa*{f(\varphi(p)),q}\).

  Now for the proof of \eqref{eta-ext}: Let \((P,\varphi) : \lifting_{\V}(X)\) be
  arbitrary and we calculate that
  \begin{align*}
    \eta^\#(P,\varphi) &\equiv
    \pa*{\Sigma_{p : P}\isdefined(\eta(\varphi(p))),
      \lambdadot{(p,q)}{\liftvalue(\eta(\varphi(p)),q)}} \\
    &\equiv
    \pa*{P \times \One,
      \lambdadot{(p,q)}{\varphi(p)}} \\
    &= (P,\varphi),
  \end{align*}
  where the final equality is seen to hold using \cref{lifting-equality}.  For
  \eqref{ext-eta}, let \(x : X\) and \(f : X \to \lifting_{\V}(Y)\) be arbitrary
  and observe that
  \begin{align*}
    f^\#(\eta(x))
    &\equiv f^\#(\One,\lambdadot{u}{x}) \\
    &\equiv \pa*{\One \times \isdefined(f(x)),
      \lambdadot{(u,p)}{\liftvalue(f(x),p)}} \\
    &= \pa*{\isdefined(f(x)),
      \lambdadot{p}{\liftvalue(f(x),p)}} \\
    &\equiv f(x)
  \end{align*}
  where the penultimate equality is another easy application of
  \cref{lifting-equality}.
  We see that these proofs amount to the fact that \(\One\) is the unit for
  taking the product of types.
  For \eqref{comp-ext} the proof amounts to the associativity of \(\Sigma\).
\end{proof}

\begin{remark}
  It should be noted that if \(X : \U\), then
  \(\lifting_{\V}(X) : \V^+ \sqcup \U\), so in general the lifting is a monad
  ``across universes''.
  But this increase in universes does not hinder us in stating and proving the
  monad laws and using them in later proofs.
\end{remark}
\begin{remark}
  The equalities of \cref{lifting-is-monad} do not include any coherence
  conditions which may be needed when \(X\) is not a set but a higher type. We
  will restrict to the lifting of sets, but the more general case is considered
  in~\cite{Escardo2021} where the coherence conditions are not needed for its
  goals either.
\end{remark}

\begin{definition}[Lifting functor, \(\lifting_{\V}(f)\)]\label{def:lifting-functor}
  The functorial action of the lifting could be defined from the unit and
  Kleisli extension as \(\lifting_{\V}(f) \colonequiv \pa*{\eta_Y \circ f}^\#\)
  for \(f : X \to Y\).
  But it is equivalent and easier to define \(\lifting_{\V}(f)\) directly by
  post-composition:
  \[
    \lifting_{\V}(f)(P,\varphi) \colonequiv \pa{P,f\circ\varphi}.
    \qedhere
  \]
\end{definition}

We now work towards showing that \(\lifting_{\V}(X)\) is the free pointed
\(\V\)-dcpo on a set~\(X\).

\begin{proposition}\label{lifting-order}
  The relation
  \({\below} : \lifting_{\V}(X) \to \lifting_{\V}(X) \to \V^+\sqcup \U\)
  given by
  \[
    l \below m \colonequiv {\isdefined(l) \to l = m}
  \]
  is a partial order on \(\lifting_{\V}(X)\) for every set \(X : \U\).
  Moreover, it is equivalent to the more verbose relation
  \[
    (P,\varphi) \below' (Q,\psi) \colonequiv \Sigma_{f : P \to Q}%
    \pa*{\varphi \sim \psi \circ f}
  \]
  that is valued in \(\V \sqcup \U\).
\end{proposition}
\begin{proof}
  Note that \({\below}\) is subsingleton-valued because \(X\) is assumed to be a
  set. The other properties follow using \cref{lifting-equality}.
\end{proof}

In light of~\cref{universe-levels-of-lifting-and-exponentials}, we carefully
keep track of the universe parameters of the lifting in the following
proposition.

\begin{proposition}[cf.~{\cite[Theorem~1]{EscardoKnapp2017}}]%
  \label{lifting-is-pointed-dcpo}%
  For a set \(X : \U\), the lifting \(\lifting_{\V}(X)\) ordered as in
  \cref{lifting-order} is a pointed \(\V\)-dcpo.
  In full generality we have \(\lifting_{\V}(X) : \DCPO{V}{V^+ \sqcup U}{V^+ \sqcup U}\), but
  if \(X : \V\), then \(\lifting_{\V}(X)\) is locally small.
\end{proposition}
\begin{proof}
  By \cref{lifting-order} we have a poset and it is clear that \(\bot\) from
  \cref{def:lifting-bot} is its least element.
  Now let \(\pa*{Q_{(-)},\varphi_{(-)}} : I \to \lifting_{\V}(X)\) be a directed
  family with \(I : \V\). We claim that the map
  \({\Sigma_{i : I}Q_i} \xrightarrow{(i,q) \mapsto \varphi_i(q)} X\) is
  constant.
  Indeed, given \(i,j : I\) with \(p : Q_i\) and \(q : Q_j\), there exists
  \(k : I\) such that
  \(\pa*{Q_i,\varphi_i},\pa*{Q_j,\varphi_j} \below \pa*{Q_k,\varphi_k}\) by
  directedness of the family.
  But by definition of the order and the elements \(p : Q_i\) and \(q : Q_j\),
  this implies that
  \(\pa*{Q_i,\varphi_i} = \pa*{Q_j,\varphi_j} = \pa*{Q_k,\varphi_k}\) which in
  particular tells us that \(\varphi_i(p) = \varphi_j(q)\).
  Hence, by~\cite[Theorem~5.4]{KrausEtAl2017}, we have a (dashed) map~\(\psi\)
  making the diagram
  \[
    \begin{tikzcd}
      \Sigma_{i : I}Q_i \ar[dr,"{\tosquash{-}}"']
      \ar[rr,"{(i,q) \mapsto \varphi_i(q)}"] & & X \\
      & \exists_{i : I}Q_i \ar[ur,dashed,"\psi"']
    \end{tikzcd}
  \]
  commute.
  We claim that \(\pa*{\exists_{i : I}Q_i,\psi}\) is the least upper bound of
  the family.
  To see that it is an upper bound, let \(j : I\) be arbitrary. By the
  commutative diagram and \cref{lifting-order} we see that
  \(\pa*{Q_j,\varphi_j} \below \pa*{\exists_{i : I}Q_i,\psi}\), as desired.
  Moreover, if \((P,\rho)\) is an upper bound for the family, then
  \(\pa*{Q_i,\varphi_i} = (P,\rho)\) for all \(i : I\) such that \(Q_i\) holds.
  Hence, \(\pa*{\exists_{i : I}Q_i,\psi} \below (P,\rho)\), as desired.
  Finally, local smallness in the case that \(X\) is a type in \(\V\) follows
  from \cref{lifting-order}.
\end{proof}

\begin{proposition}\label{lifting-extension-is-continuous}%
  The Kleisli extension \(f^\# : \lifting_{\V}(X) \to \lifting_{\V}(Y)\) is
  Scott continuous for any map \(f : X \to \lifting_{\V}(Y)\).
\end{proposition}
\begin{proof}
  It is straightforward to prove that \(f^\#\) is monotone. Hence, it remains to
  prove that \(f^\#\pa*{\bigsqcup \alpha} \below \bigsqcup f^\# \circ \alpha\)
  for every directed family \(\alpha : I \to \lifting_{\V}(X)\).
  So suppose that \(f^\#\pa*{\bigsqcup \alpha}\) is defined. Then we have to
  show that it equals \(\bigsqcup f^\# \circ \alpha\).
  By our assumption and definition of \(f^\#\) we get that \(\bigsqcup \alpha\)
  is defined too. By the definition of suprema in the lifting and because we are
  proving a proposition, we may assume to have \(i : I\) such that \(\alpha_i\)
  is defined.
  But since \(\alpha_i \below \bigsqcup \alpha\), we get
  \(\alpha_i = \bigsqcup \alpha\) and hence,
  \(f^\#\pa*{\alpha_i} = f^\#\pa*{\bigsqcup \alpha}\).
  Finally, \(f^\#\pa*{\alpha_i} \below \bigsqcup f^\# \circ \alpha\), but by
  assumption \(f^\#\pa*{\bigsqcup \alpha}\) is defined and hence so is
  \(f^\#\pa*{\alpha_i}\) which implies
  \(f^\#\pa*{\bigsqcup \alpha} = f^\#\pa*{\alpha_i} = \bigsqcup f^\# \circ
  \alpha\), as desired.
\end{proof}

Recall from \cref{pointed-dcpos-sups} that pointed \(\V\)-dcpos have suprema of
families indexed by propositions in \(\V\). We make use of this fact in the
following lemma.

\begin{lemma}\label{lifting-element-as-sup}
  For a set \(X\), every partial element \((P,\varphi) : \lifting_{\V}(X)\) is
  equal to supremum \(\bigvee_{p : P}\eta_X\pa*{\varphi(p)}\).
\end{lemma}
\begin{proof}
  Note that if \(p : P\), then \((P,\varphi) = \eta_X(\varphi(p))\), so that the
  lemma follows from antisymmetry.
\end{proof}

The lifting \(\lifting_{\V }(X)\) gives the free pointed \(\V \)-dcpo on a set
\(X\). Keeping track of universes, it holds in the following generality:
\begin{theorem}\label{lifting-is-free}%
  If \(X : \U \) is a set, then for every pointed \(\V \)-dcpo
  \(D : \DCPO{V}{U'}{T'}\) and function \(f : X \to D\), there is a unique
  strict continuous function \(\bar{f} : \lifting_{\V }(X) \to D\) making
  the diagram
  \[
    \begin{tikzcd}
      X \ar[dr, "\eta_X"'] \ar[rr, "f"] & & D  \\
      & \lifting_{\V }(X) \ar[ur, dashed, "\bar{f}"']
    \end{tikzcd}
  \]
  commute.
\end{theorem}
\begin{proof}
  We define \(\bar f : \lifting_{\V }(X) \to D\) by
  \((P,\varphi) \mapsto \bigvee_{p : P} f(\varphi(p))\), which is well-defined
  by \cref{pointed-dcpos-sups} and easily seen to be strict
  and continuous. For uniqueness, suppose that we have
  \({g : \lifting_{\V }(X) \to D}\) strict and continuous such that
  \(g \circ \eta_X = f\) and let \((P,\varphi)\) be an arbitrary element of
  \(\lifting_{\V }(X)\). Then,
  \begin{align*}
    g\pa*{P,\varphi}
    &= g\pa*{\textstyle\bigvee_{p : P}\eta_X\pa*{\varphi(p)}}
    &&\text{(by \cref{lifting-element-as-sup})} \\
    &= \textstyle\bigvee_{p : P}g\pa*{\eta_X\pa*{\varphi(p)}}
    &&\text{(by \cref{pointed-dcpos-sups} and strictness and continuity of \(g\))} \\
    &= \textstyle\bigvee_{p : P} f\pa*{\varphi(p)}
    &&\text{(by assumption on \(g\))} \\
    &\equiv \bar{f}(P,\varphi),
  \end{align*}
  as desired.
\end{proof}

Finally, we consider a variation of \cref{lifting-order} which allows us to
freely add a least element to a \(\V\)-dcpo instead of just a set.

\begin{proposition}\label{lifting-order-alt}
  For a poset \(D\) whose order takes values in \(\T\), the binary relation
  \({\below} : \lifting_{\V}(D) \to \lifting_{\V}(D) \to \V \sqcup \T\) %
  given by
  \[
    (P,\varphi) \below (Q,\psi) \colonequiv \Sigma_{f : P \to Q}%
    \pa*{\Pi_{p : P}\pa*{\varphi(p) \below_D \psi(f(p))}}
  \]
  is a partial order on \(\lifting_{\V}(D)\).
\end{proposition}
\begin{proof}
  Similar to \cref{lifting-order}, but using that \({\below_D}\) is reflexive,
  transitive and antisymmetric.
\end{proof}

\begin{proposition}\label{lifting-is-pointed-dcpo-alt}%
  For a dcpo \(D : \DCPO{V}{U}{T}\), the lifting \(\lifting_{\V}(D)\) ordered as
  in \cref{lifting-order-alt} is a pointed \(\V\)-dcpo.
  In general, \(\lifting_{\V}(D) : \DCPO{V}{V^+ \sqcup U}{V \sqcup T}\), but if
  \(D\) is locally small, then so is \(\lifting_{\V}(D)\).
\end{proposition}
\begin{proof}
  The element \(\bot\) from \cref{def:lifting-bot} is still the least element
  with respect to the new order.
  If \(\alpha : I \to \lifting_{\V}(D)\) is directed, then we consider
  \(\Phi : \pa*{\Sigma_{i : I}Q_i} \to D\) given by
  \((i,q) \mapsto \varphi_i(q)\).
  This family is semidirected, for if we have \(i,j : I\) with \(p : Q_i\) and
  \(q : Q_j\), then there exists \(k : I\) such that
  \({\alpha_i,\alpha_j} \below \alpha_k\) in \(\lifting_{\V}(D)\) by
  directedness of \(\alpha\), which implies that \(\Phi(i,p) \below \Phi(j,q)\)
  in \(D\).
  Thus, if we know that \(\exists_{i : I}Q_i\), then the family \(\Phi\) is
  directed and must have a supremum in \(D\).
  Hence we have a partial element
  \(\pa*{\exists_{i : I}Q_i,\psi} : \lifting_{\V}(D)\) where \(\psi\) takes the
  witness that the domain of \(\Phi\) is inhabited to the directed supremum
  \(\bigsqcup \Phi\) in \(D\).
  It is not hard to verify that this partial element is the least upper bound of
  \(\alpha\) in \(\lifting_{\V}(D)\), completing the proof.
\end{proof}

The lifting \(\lifting_{\V }(D)\) with the partial order of
\cref{lifting-order-alt} gives the free \emph{pointed} \(\V\)-dcpo on a
\(\V\)-dcpo \(D\). Keeping track of universes, it holds in the following
generality:
\begin{theorem}\label{lifting-is-free3}
  If \(D : \DCPO{V}{U}{T}\) is a \(\V\)-dcpo, then for every pointed \(\V\)-dcpo
  \(E : \DCPO{V}{U'}{T'}\) and continuous function \(f : D \to E\), there is a
  unique strict continuous function \(\bar{f} : \lifting_{\V }(D) \to E\) making
  the diagram
  \[
    \begin{tikzcd}
      D \ar[dr, "\eta_{D}"'] \ar[rr, "f"] & & E  \\
      & \lifting_{\V }(D) \ar[ur, dashed, "\bar{f}"']
    \end{tikzcd}
  \]
  commute.
\end{theorem}
\begin{proof}
  Similar to the proof of \cref{lifting-is-free}.
\end{proof}

Notice how \cref{lifting-is-free3} generalises \cref{lifting-is-free} as any set
can be viewed as a discretely ordered \(\V\)-dcpo.

\section{Products and exponentials}\label{sec:products-and-exponentials}

We describe two constructions of \(\V\)-dcpos, namely products and exponentials.
Exponentials will be crucial in Scott's \(D_\infty\) construction, as discussed
in~\cref{sec:Scott-D-infty}.
Products are not needed for this purpose, but they allow us to state the
universal property of the exponential (\cref{exponential-universal-property}).

\begin{definition}[Product of (pointed) dcpos, \(D_1 \times D_2\)]
  The \emph{product} of two \(\V\)-dcpos \(D_1\) and \(D_2\) is given by the
  \(\V\)-dcpo \(D_1 \times D_2\) defined as follows. Its carrier is the
  cartesian product of the carriers of \(D_1\) and \(D_2\). The order is given
  componentwise, i.e.\ \((x,y) \below_{D_1 \times D_2} (x',y')\) if
  \(x \below_{D_1} y\) and \(y \below_{D_2} y'\).
  Accordingly directed suprema are also given componentwise. That is, given a
  directed family \(\alpha : I \to {D_1 \times D_2}\), one quickly verifies that
  the families \(\fst \circ {\alpha}\) and \(\snd \circ \alpha\) are also
  directed. We then define the supremum of \(\alpha\) as
  \((\bigsqcup {\fst \circ \alpha} , \bigsqcup {\snd \circ \alpha})\).
  Moreover, if \(D\) and \(E\) are pointed, then so is \(D \times E\) by taking
  the least elements in both components.
\end{definition}

\begin{remark}
  Notice that if \(D_1 : \DCPO{V}{U}{T}\) and \(D_2 : \DCPO{V}{U'}{T'}\), then
  for their product we have
  \(D_1 \times D_2 : \DCPO{V}{U \sqcup U'}{T \sqcup T'}\), which simplifies to
  \(\DCPO{V}{U}{T}\) when \(\U' \equiv \U\) and \(\T' \equiv \T\).
\end{remark}

\begin{proposition}\label{binary-product-universal-property}
  The product defined above satisfies the appropriate universal property: the
  projections \(\fst : D_1 \times D_2 \to D_1\) and
  \(\snd : D_1 \times D_2 \to D_2\) are Scott continuous and if
  \(f : E \to D_1\) and \(g : E \to D_2\) are Scott continuous functions from a
  \(\V\)-dcpo~\(E\), then there is a unique Scott continuous map
  \(k \colon E \to D_1 \times D_2\) such that the diagram
  \[
    \begin{tikzcd}
      & D_1 \times D_2 \ar[dl,"\fst"'] \ar[dr,"\snd"] \\
      D_1 \ & & D_2 \\
      & E \ar[ul,"f"] \ar[ur,"g"'] \ar[uu,dashed,"k"]
    \end{tikzcd}
  \]
  commutes.
\end{proposition}
\begin{proof}
  The projections are Scott continuous by definition of directed suprema in
  \(D_1 \times D_2\). Moreover, if \(f \colon E \to D_1\) and
  \(g \colon E \to D_2\) are Scott continuous maps, then we see that we have no
  choice but to define \(k \colon E \to {D_1 \times D_2}\) by
  \(e \mapsto (f(e),g(e))\). Moreover, this assignment is Scott continuous,
  because for a directed family \(\alpha : I \to E\), we have
  \(k\pa*{\bigsqcup \alpha} \equiv \pa*{f(\bigsqcup \alpha),g(\bigsqcup\alpha)}
  = \pa*{\bigsqcup f \circ \alpha,\bigsqcup g \circ \alpha} \equiv \bigsqcup {k
    \circ \alpha}\) by Scott continuity of \(f\)~and~\(g\) and the definition of
  directed suprema in \(D_1 \times D_2\).
\end{proof}

\begin{lemma}\label{continuous-in-both-arguments}
  A map \(f \colon D_1 \times D_2 \to E\) is Scott continuous if and only if the
  maps \(f(x,-) : D_2 \to E\) and \(f(-,y) : D_1 \to E\) are Scott continuous
  for every \(x : D_1\) and \(y : D_2\).
\end{lemma}
\begin{proof}
  Straightforward as in the classical case.
\end{proof}

\begin{definition}[Exponential of (pointed) dcpos, \(E^D\)]
  The \emph{exponential} of two \(\V\)\nobreakdash-dcpos \(D\) and \(E\) is
  given by the poset \(E^D\) defined as follows. Its carrier is the type of
  Scott continuous functions from \(D\) to \(E\).
  The order is given pointwise, i.e.\ \(f \below_{E^D} g\) holds if
  \(f(x) \below_{E} g(x)\) for every \(x : D\).
  Notice that if \(E\) is pointed, then so is \(E^D\) with least element given
  the constant function \(\lambdadot{x : D}\bot_E\) which is Scott continuous by
  \cref{continuity-closure}\eqref{const-is-continuous}.
  Finally, it is straightforward to show that \(E^D\) is \(\V\)-directed
  complete, so that \(E^D\) is another \(\V\)-dcpo.
\end{definition}

Note that the exponential \(E^D\) is a priori not locally small even if \(E\) is
because the partial order quantifies over all elements of \(D\). But if \(D\) is
continuous (a notion that we study in detail in the Part~II paper
\cite{deJongEscardoCompanion}) then \(E^D\) will be locally small
when \(E\) is.

\begin{remark}\label{exponential-universe-parameters}%
  Recall from~\cref{universe-levels-of-lifting-and-exponentials} that it is
  necessary to carefully keep track of the universe parameters of the
  exponential.
  In general, the universe levels of \(E^{D}\) can be quite large and
  complicated. For~if \(D : \DCPO{V}{U}{T}\) and \(E : \DCPO{V}{U'}{T'}\), then
  the exponential \(E^D\) has a carrier in the universe
  \[
    \V^+ \sqcup \U \sqcup \T \sqcup \U' \sqcup \T'
  \]
  and an order relation that
  takes values in
  \(
    \U \sqcup \T'.
    \)

  Even~if
  \(\V = \U \equiv \T \equiv \U' \equiv \T'\), the carrier of \(E^{D}\) still
  lives in the larger universe~\(\V ^+\), because the type expressing Scott
  continuity for \(\V\)-dcpos quantifies over all types in~\(\V\).
  Actually, the scenario where \(\U = \U' = \V\) cannot happen in a predicative
  setting unless \(D\) and \(E\) are trivial, in a sense made precise in
  \cite{deJongEscardo2023}.

  Even so, in many applications such as those in \cite{deJong2021a}
  or~\cref{sec:Scott-D-infty}, if we take \(\V \equiv \U_0\) and all other
  parameters to be \(\U \equiv \T \equiv \U' \equiv \T' \equiv \U_1\), then the
  situation is much simpler and \(D\),~\(E\)~and the exponential \(E^D\) are all
  elements of \(\DCPOnum{0}{1}{1}\) with all of them being locally small
  (remember that this is defined up to equivalence).
  This turns out to be a very favourable situation for both the Scott model of
  PCF~\cite{deJong2021a} and Scott's \(D_\infty\) model of the untyped
  \(\lambda\)-calculus~(\cref{sec:Scott-D-infty}).
\end{remark}

In the proposition below we can have \(D : \DCPO{\V}{\U}{\T}\) and
\(E : \DCPO{\V}{\U'}{\T'}\) for arbitrary universes \(\U\), \(\T\), \(\U'\) and
\(\T'\). In particular, the universe parameters of \(D\) and \(E\), apart from
the universe of indexing types, need not be the same.

\begin{proposition}\label{exponential-universal-property}%
  The exponential defined above satisfies the appropriate universal property:
  the \emph{evaluation map} \(\ev: E^D \times D \to E, (g,x) \mapsto g(x)\) is
  Scott continuous and if \(f : {D'\times D} \to E\) is a Scott continuous
  function, then there is a unique Scott continuous map
  \(\bar{f} \colon D' \to E^D\) such that the diagram
  \[
    \begin{tikzcd}
      D' \times D \ar[dr,"f"] \ar[d,dashed,"{\bar{f}}\,\times\,{\id_D}"'] \\
      E^D \times D \ar[r,"\ev"'] & E
    \end{tikzcd}
  \]
  commutes.
\end{proposition}
\begin{proof}
  Straightforward application of \cref{continuous-in-both-arguments}, as in the
  classical case.
\end{proof}

For completeness, we recall the following theorem which lies at the heart of the
Scott model of PCF studied in~\cite{deJong2021a}:

\begin{theorem}[Least fixed point, \(\mu\)]\label{least-fixed-point}%
  Every Scott continuous endomap \(f\) on a pointed \(\V\)-dcpo \(D\) has a
  least fixed point given by
  \[
    \mu(f) \colonequiv \textstyle\bigsqcup_{n : \Nat} f^n(\bot).
  \]
  Specifically, the following two conditions hold:
  \begin{enumerate}[(i)]
  \item\label{is-fixed-point} \(f(\mu(f)) = \mu(f)\), and
  \item\label{is-least-fixed-point} for every \(x : D\), if \(f(x) \below x\), then
    \(\mu(f) \below x\).
  \end{enumerate}
  Moreover, the assignment \(f \mapsto \mu(f)\) defines a Scott continuous map
  \(D^D \to D\).
\end{theorem}
\begin{proof}
  Following the proof~\cite[Theorem~2.1.19]{AbramskyJung1994};
  see~\cite[Theorem~16]{deJong2021a} for details.
\end{proof}

\section{Bilimits}\label{sec:bilimits}

Recall that in a \(\V\)-dcpo \(D\) we can take suprema of directed families
\(\alpha : I \to D\). It is a striking feature of directed complete posets that
this act is reflected in the \emph{category} of dcpos, although it does require
us to specify an appropriate notion of one dcpo being ``below'' another one.
This notion will be exactly that of an embedding-projection pair.
The technical results developed in this section will find application in the
construction of Scott's \(D_\infty\), a model of the untyped
\(\lambda\)-calculus, as discussed in \cref{sec:Scott-D-infty}.
We refer the reader to the conclusion (\cref{sec:conclusion}) for a discussion
on how our predicative reconstruction of \(D_\infty\) compares Scott's original
work~\cite{Scott1972}.

\begin{definition}[Deflation]
  An endofunction \(f : D \to D\) on a poset \(D\) is a \emph{deflation} if
  \(f(x) \below x\) for all \(x : D\).
\end{definition}
\begin{definition}[Embedding-projection pair]%
  An \emph{embedding-projection pair} from a \(\V\)-dcpo \(D\) to a \(\V\)-dcpo
  \(E\) consists of two Scott continuous functions \(\varepsilon : D \to E\)
  (the~\emph{embedding}) and \(\pi : E \to D\) (the~\emph{projection})
  such~that:
  \begin{enumerate}[(i)]
  \item \(\varepsilon\) is a section of \(\pi\), and
  \item \(\varepsilon \circ \pi\) is a deflation.%
    \qedhere
  \end{enumerate}
\end{definition}

For the remainder of this section, fix the following setup, where we try to be
as general regarding universe levels as we can be.
We fix a directed preorder \((I,\below)\) with \(I : \V\) and \(\below\) takes
values in some universe \(\W\). Now suppose that \((I,\below)\) indexes a family
of \(\V\)-dcpos with embedding-projection pairs between them, i.e.\ we have
\begin{itemize}
\item for every \(i : I\), a \(\V \)-dcpo \(D_i : \DCPO{V}{U}{T}\), and
\item for every \(i,j : I\) with \(i \sqsubseteq j\), an embedding-projection
  pair \(\pa*{\varepsilon_{i,j},\pi_{i,j}}\) from \(D_i\) to \(D_j\).
\end{itemize}
Moreover, we require that the following compatibility conditions hold:
\begin{align}
  &\text{for every \(i : I\), we have \(\varepsilon_{i,i} = \pi_{i,i} = \id\)}; %
  \label{epsilon-pi-id} \\
  &\text{for every \(i \sqsubseteq j \sqsubseteq k\) in \(I\), we have
  \(\varepsilon_{i,k} \sim \varepsilon_{j,k} \circ \varepsilon_{i,j}\) and
    \(\pi_{i,k} \sim \pi_{i,j} \circ \pi_{j,k}\).} %
  \label{epsilon-pi-comms}
\end{align}

\begin{example}
  If \(I \colonequiv \Nat\) with the usual ordering, then we are looking at a
  diagram of \(\V\)-dcpos like this
  \[
    \begin{tikzcd}
      D_0 \ar[r,"\varepsilon_{0,1}", shift left, hookrightarrow] &
      D_1 \ar[l,"\pi_{0,1}", shift left, two heads]
      \ar[r,"\varepsilon_{1,2}", shift left, hookrightarrow ]
      &
      D_2 \ar[l,"\pi_{1,2}", shift left, two heads]
      \ar[r,"\varepsilon_{2,3}", shift left, hookrightarrow ]
      &
      D_3 \ar[l,"\pi_{2,3}", shift left, two heads]
      \ar[r,"\varepsilon_{3,4}", shift left, hookrightarrow ]
      &
      \cdots
      \ar[l,"\pi_{3,4}", shift left, two heads]
    \end{tikzcd}
  \]
  where, for example, we have not pictured
  \(\varepsilon_{1,1} : D_1 \hookrightarrow D_1\) and
  \(\varepsilon_{0,2} : D_0 \hookrightarrow D_2\) explicitly, as they are equal
  to \(\id_{D_1} : D_1 \to D_1\) and the composition of
  \(D_0 \xhookrightarrow{\varepsilon_{0,1}} D_1\) and
  \(D_1\xhookrightarrow{\varepsilon_{1,2}} D_2\), respectively.
\end{example}

The goal is now to construct another \(\V\)-dcpo \(D_\infty\) with
embedding-projections pairs
\(\pa*{\varepsilon_{i,\infty} : D_1 \hookrightarrow D_\infty, {\pi_{i,\infty} :
    D_\infty \to D_i}}\) for every \(i : I\), such that
\(\pa*{D_\infty,\pa*{\varepsilon_{i,\infty}}_{i : I}}\) is the colimit of the
diagram given by \(\pa*{\varepsilon_{i,j}}_{i \below j \text{ in } I}\) and
\(\pa*{D_\infty,\pa*{\pi_{i,\infty}}_{i : I}}\) is the limit of the
diagram given by \(\pa*{\pi_{i,j}}_{i \below j \text{ in } I}\).
In other words,
\(\pa*{D_\infty,\pa*{\varepsilon_{i,\infty}}_{i : I},\pa*{\pi_{i,\infty}}_{i :
    I}}\) is both the colimit and the limit in the category of \(\V\)-dcpos with
embedding-projections pairs between them. We say that it is the \emph{bilimit}.

\begin{definition}[\(D_\infty\)]\label{def:D-infty}%
  We define a poset \(D_\infty\) as follows. Its carrier is given by dependent
  functions \(\sigma : \Pi_{i : I}D_i\) satisfying
  \(\pi_{i,j}(\sigma_j) = \sigma_i\) whenever \(i \below j\).
  That is, the carrier is the type
  \[
    \sum_{\sigma : \Pi_{i : I} D_i} \prod_{{i,j} : I , i \below j}
    \pi_{i,j}\pa*{\sigma_j} = \sigma_i.
  \]
  Note that this defines a subtype of \(\Pi_{i : I}D_i\) as the condition
  \(\prod_{{i,j} : I , i \below j} \pi_{i,j}\pa*{\sigma_j} = \sigma_i\) is a
  property by \cite[Example~3.6.2]{HoTTBook} and the fact that each \(D_i\) is a
  set.

  These functions are ordered pointwise, i.e.\ if
  \(\sigma,\tau : \Pi_{i : I} D_i\), then \(\sigma \below_{D_\infty} \tau\)
  exactly when \(\sigma_i \below_{D_i} \tau_i\) for every \(i : I\).
\end{definition}

\begin{lemma}
  The poset \(D_\infty\) is \(\V\)-directed complete with suprema calculated
  pointwise.
  Paying attention to the universe levels involved, we have
  \(D_\infty : \DCPO{V}{U \sqcup V \sqcup W}{U \sqcup T}\).
\end{lemma}
\begin{proof}
  If \(\alpha : A \to D_\infty\) is a directed family, then the family
  \(\alpha_i : A \to D_i\) given by
  \(\alpha_i(a) \colonequiv \pa*{\alpha(a)}_i\) is directed again, and we define
  the supremum of \(\alpha\) in \(D_\infty\) as the function
  \(i \mapsto \bigsqcup \alpha_i\).
  To see that this indeed defines an element of \(D_\infty\), observe that for
  every \(i,j : I\) with \(i \below j\) we have
  \begin{align*}
    \pi_{i,j}\big(\pa*{\textstyle\bigsqcup \alpha}_j\big)
    &\equiv \pi_{i,j}\pa*{\textstyle\bigsqcup\alpha_j}
    \\ &= \textstyle\bigsqcup \pi_{i,j} \circ \alpha_j
       &&\text{(by Scott continuity of \(\pi_{i,j}\))}
    \\ &\equiv \textstyle\bigsqcup_{a : A}\big(\pi_{i,j}\big(\pa*{\alpha(a)}_j\big)\big)
    \\ &= \textstyle\bigsqcup \alpha_i
       &&\text{(as \(\alpha(a)\) is an element of \(D_\infty\))},
  \end{align*}
  as desired.
\end{proof}

\begin{remark}\label{bilimit-universe-parameters}%
  We allow for general universe levels here, which is why \(D_\infty\) lives in
  the relatively complicated universe \(\U \sqcup \V \sqcup \W\). In concrete
  examples, the situation often simplifies. E.g., in \cref{sec:Scott-D-infty} we
  find ourselves in the favourable situation described
  in~\cref{exponential-universe-parameters} where \(\V \equiv \W \equiv \U_0\)
  and \(\U \equiv \T \equiv \U_1\), so that we get
  \(D_\infty : \DCPOnum{0}{1}{1}\), as the bilimit of a diagram of dcpos
  \(D_n : \DCPOnum{0}{1}{1}\) indexed by natural numbers.
\end{remark}

\begin{definition}[\(\pi_{i,\infty}\)]\label{pi-infty}
  For every \(i : I\), we define the Scott continuous function
  \(\pi_{i,\infty} : {D_\infty \to D_i}\) by \(\sigma \mapsto \sigma_i\).
\end{definition}

\begin{lemma}\label{pi-infty-is-continuous}
  The map \(\pi_{i,\infty} : D_\infty \to D_i\) is Scott continuous for every
  \(i : I\).
\end{lemma}
\begin{proof}
  This holds because suprema in \(D_\infty\) are calculated pointwise and
  \(\pi_{i,\infty}\) selects the \(i\)-th component.
\end{proof}

While we could closely follow~\cite{Scott1972} up until this point, we will now
need a new idea to proceed.
Our goal is to define maps \(\varepsilon_{i,\infty} : D_i \to D_\infty\) for
every \(i : I\) so that \(\varepsilon_{i,\infty}\) and \(\pi_{i,\infty}\) form
an embedding-projection pair.
We give an outline of the idea for defining this map
\(\varepsilon_{i,\infty}\). For~an arbitrary element \(x : D_i\), we need to
construct \(\sigma : D_\infty\) at component \(j : I\), say. If we had \(k : I\)
such that \(i,j \below k\), then we could define \(\sigma_j : D_j\) by
\(\pi_{j,k}\pa*{\varepsilon_{i,k}(x)}\).
Now semidirectedness of \(I\) tells us that there exists such a \(k : I\), so
the point is to somehow make use of this propositionally truncated fact. This is
where~\cite[Theorem~5.4]{KrausEtAl2017} comes in.
We define a map
\(\kappa_{i,j}^x : \pa*{\Sigma_{k : I}\,\pa{i \below k} \times \pa{j \below k}}
\to D_j\) by sending \(k\) to \(\pi_{j,k}\pa*{\varepsilon_{i,k}(x)}\) and show
it to be constant, so that it factors through the truncation of its domain.
In the special case that \(I \equiv \Nat\), as in~\cite{Scott1972}, we could
simply take \(k\) to be the sum of the natural numbers \(i\) and \(j\), but this
does not work in the more general directed case, of course.

\begin{definition}[\(\kappa_{i,j}^x\)]\label{def:kappa}
  For every \(i,j : I\) and \(x : D_i\) we define the function
  \[
    \kappa_{i,j}^x : \pa*{\Sigma_{k : I}\,\pa{i \below k} \times \pa{j \below
        k}} \to D_j
  \]
  by mapping \(k : I\) with \(i,j \below k\) to
  \(\pi_{j,k}\pa*{\varepsilon_{i,k}(x)}\).
\end{definition}

\begin{lemma}\label{kappa-is-constant}
  The function \(\kappa_{i,j}^x\) is constant for ever \(i,j : I\) and
  \(x : D_i\).
  Hence, \(\kappa_{i,j}^x\) factors through
  \(\exists_{k : I}\,\pa{i \below k}\times\pa{j \below k}\)
  by~\cite[Theorem~5.4]{KrausEtAl2017}.
\end{lemma}
\begin{proof}
  If we have \(k_1,k_2 : I\) with \(i \below k_1,k_2\) and \(j \below k_1,k_2\),
  then by semidirectedness of \(I\), there exists some \(k : K\) with
  \(k_1,k_2 \below k\) and hence,
  \begin{align*}
    \pa*{\pi_{j,k_1} \circ \varepsilon_{i,k_1}} (x)
    &= \pa*{\pi_{j,k_1} \circ \pi_{k_1,k} \circ \varepsilon_{k_1,k} \circ \varepsilon_{i,k_1}} (x)
      &&\text{(since \(\varepsilon_{k_1,k}\) is a section of \(\pi_{k_1,k}\))}
    \\
    &= \pa*{\pi_{j,k} \circ \varepsilon_{i,k}}(x)
      &&\text{(by \cref{epsilon-pi-comms})}
    \\
    &= \pa*{\pi_{j,k} \circ \pi_{k_2,k} \circ \varepsilon_{k_2,k} \circ \varepsilon_{i,k_2}} (x)
      &&\text{(since \(\varepsilon_{k_2,k}\) is a section of \(\pi_{k_2,k}\))}
    \\
    &= \pa*{\pi_{j,k_2} \circ \varepsilon_{i,k_2}}(x)
      &&\text{(by \cref{epsilon-pi-comms})},
  \end{align*}
  proving that \(\kappa_{i,j}^x\) is constant.
\end{proof}

\begin{definition}[\(\rho_{i,j}\)]
  For every \(i,j : I\), the type
  \(\exists_{k : I}\,\pa{i \below k} \times \pa{j \below k}\) has an element
  since \((I,\below)\) is directed. Thus, \cref{kappa-is-constant} tells us that
  we have a function \(\rho_{i,j} : D_i \to D_j\) such that if \(i,j \below k\),
  then the equation
  \begin{equation}\label{rho-eq}
    \rho_{i,j}(x) = \kappa_{i,j}^x(k) \equiv \pi_{j,k}\pa*{\varepsilon_{i,k}(x)}
  \end{equation}
  holds for every \(x : D_i\).
\end{definition}

\begin{definition}[\(\varepsilon_{i,\infty}\)]\label{epsilon-infty}
  The map \(\rho\) induces a map \(\varepsilon_{i,\infty} : D_i \to D_\infty\)
  by sending \(x : D_i\) to the function \(\lambdadot{j : I}{\rho_{i,j}(x)}\).
  To see that this is well-defined, assume that we have \(j_1 \below j_2\) in
  \(J\) and \(x : D_i\). We have to show that
  \(\pi_{j_1,j_2}\pa*{\pa*{\varepsilon_{i,\infty}(x)}_{j_2}} =
  \pa*{\varepsilon_{i,\infty}(x)}_{j_1}\).
  By semidirectedness of \(I\) and the fact that are looking to prove a
  proposition, we may assume to have \(k : I\) with \(i \below k\) and
  \(j_1 \below j_2 \below k\). Then,
  \begin{align*}
    \pi_{j_1,j_2}\pa*{\pa*{\varepsilon_{i,\infty}(x)}_{j_2}}
    &\equiv \pi_{j_1,j_2}\pa*{\rho_{i,j_2}(x)} \\
    &= \pi_{j_1,j_2}\pa*{\pi_{j_2,k}\pa*{\varepsilon_{i,k}(x)}}
    &&\text{(by \cref{rho-eq})} \\
    &= \pi_{j_1,k}\pa*{\varepsilon_{i,k}(x)}
    &&\text{(by \cref{epsilon-pi-comms})}
    \\
    &= \rho_{i,j_1}(x)
    &&\text{(by \cref{rho-eq})} \\
    &\equiv \pa*{\varepsilon_{i,\infty}(x)}_{j_1},
  \end{align*}
  as desired.
\end{definition}

This completes the definition of \(\varepsilon_{i,\infty}\). From this point on,
we can typically work with it by using~\cref{rho-eq} and the fact that
\(\pa*{\varepsilon_{i,\infty}(x)}_j\) is defined as \(\rho_{i,j}(x)\).

\begin{lemma}\label{rho-is-continuous}
  The map \(\rho_{i,j} : D_i \to D_j\) is Scott continuous for every \(i,j : I\).
\end{lemma}
\begin{proof}
  Since we are proving a property, we may use semidirectedness of \(I\) to get
  \(k : I\) with \(i,j \below k\). Then,
  \(\rho_{i,j} \sim \pi_{j,k} \circ \varepsilon_{i,k}\) by \cref{rho-eq}. But
  the functions \(\pi_{j,k}\) and \(\varepsilon_{i,k}\) are continuous and
  continuity is preserved by function composition, so \(\rho_{i,j}\) is
  continuous, as we wished to show.
\end{proof}

\begin{lemma}\label{epsilon-infty-is-continuous}
  The map \(\varepsilon_{i,\infty} : D_i \to D_\infty\) is Scott continuous for
  every \(i : I\).
\end{lemma}
\begin{proof}
  If \(\alpha : A \to D_i\) is directed, then for every \(j : I\) we have
  \begin{align*}
    \pa*{\varepsilon_{i,\infty}\pa*{\textstyle\bigsqcup \alpha}}_j
    &\equiv \rho_{i,j} \pa*{\textstyle\bigsqcup \alpha} \\
    &= \textstyle\bigsqcup \rho_{i,j} \circ \alpha
    &&\text{(by \cref{rho-is-continuous})} \\
    &\equiv \textstyle\bigsqcup_{a : A} \pa*{\varepsilon_{i,\infty}\pa*{\alpha(a)}}_j \\
    &\equiv \pa*{\textstyle\bigsqcup \pa*{\varepsilon_{i,\infty} \circ \alpha}}_j
      &&\text{(as suprema in \(D_\infty\) are calculated pointwise)}.
  \end{align*}
  Hence,
  \(\varepsilon_{i,\infty}\pa*{\bigsqcup \alpha} = \bigsqcup
  \pa*{\varepsilon_{i,\infty} \circ \alpha}\) and \(\varepsilon_{i,\infty}\) is
  seen to be Scott continuous.
\end{proof}

\begin{theorem}\label{epsilon-pi-infty-ep-pair}
  For every \(i : I\), the pair
  \(\pa*{\varepsilon_{i,\infty},\pi_{i,\infty}}\) is an embedding-projection
  pair from \(D_i\) to \(D_\infty\).
\end{theorem}
\begin{proof}
  Scott continuity of both maps is given by
  \cref{pi-infty-is-continuous,epsilon-infty-is-continuous}. To see that
  \(\varepsilon_{i,\infty}\) is a section of \(\pi_{i,\infty}\), observe that
  for every \(x : D_i\), we have
  \begin{align*}
    \pi_{i,\infty}\pa*{\varepsilon_{i,\infty}(x)}
    &\equiv \pa*{\varepsilon_{i,\infty}(x)}_i \\
    &\equiv \rho_{i,i}(x) \\
    &= \pi_{i,i}\pa*{\varepsilon_{i,i}(x)}
      &&\text{(by \cref{rho-eq})}
    \\
    &\equiv x &&\text{(by \cref{epsilon-pi-id})},
  \end{align*}
  so that \(\varepsilon_{i,\infty}\) is indeed a section of \(\pi_{i,\infty}\).
  It remains to prove that
  \(\varepsilon_{i,\infty}\pa*{\pi_{i,\infty}\pa{\sigma}} \below \sigma\) for
  every \(\sigma : D_\infty\). The order is given pointwise, so let \(j : I\) be
  arbitrary and since we are proving a proposition, assume that we have
  \(k : I\) with \(i,j \below k\). Then,
  \begin{align*}
    \pa*{\varepsilon_{i,\infty}\pa*{\pi_{i,\infty}(\sigma)}}_j
    &\equiv \pa*{\varepsilon_{i,\infty}\pa{\sigma_i}}_j \\
    &\equiv \rho_{i,j}\pa*{\sigma_i} \\
    &= \pi_{j,k}\pa*{\varepsilon_{i,k}\pa*{\sigma_i}}
    &&\text{(by \cref{rho-eq})}
    \\
    &= \pi_{j,k}\pa*{\varepsilon_{i,k}\pa*{\pi_{i,k}\pa*{\sigma_k}}}
    &&\text{(since \(\sigma\) is an element of \(D_\infty\))}
    \\
    \shortintertext{But \(\pi_{i,k} \circ \varepsilon_{i,k}\) is deflationary and \(\pi_{j,k}\) is monotone, so}
    &\below \pi_{j,k}\pa*{\sigma_k}
    \\
    &= \sigma_j
    &&\text{(since \(\sigma\) is an element of \(D_\infty\))},
  \end{align*}
  finishing the proof.
\end{proof}

\begin{lemma}\label{epsilon-pi-infty-commutes}
  The maps \(\pi_{i,\infty}\) and \(\varepsilon_{i,\infty}\) respectively
  commute with \(\pi_{i,j}\) and \(\varepsilon_{i,j}\) whenever \(i \below j\),
  viz.\ the diagrams
  \[
    \begin{tikzcd}
      D_\infty \ar[dr,"\pi_{j,\infty}"'] \ar[rr,"\pi_{i,\infty}"] & & D_i \\
        & D_j \ar[ur, "\pi_{i,j}"']
    \end{tikzcd}
    \quad
    \quad
    \quad
    \begin{tikzcd}
      D_i \ar[dr,"\varepsilon_{i,j}"'] \ar[rr,"\varepsilon_{i,\infty}"] & & D_\infty \\
        & D_j \ar[ur, "\varepsilon_{j,\infty}"']
    \end{tikzcd}
  \]
  commute for all \(i,j : I\) with \(i \below j\).
\end{lemma}
\begin{proof}
  If \(i \below j\) and \(\sigma : D_\infty\) is arbitrary, then
  \[
    \pi_{i,j}\pa*{\pi_{j,\infty}\pa{\sigma}} \equiv \pi_{i,j}\pa*{\sigma_j} =
    \sigma_i
  \]
  precisely because \(\sigma\) is an element of \(D_\infty\), which proves the
  commutativity of the first diagram.
  For the second, let \(x : D_i\) be arbitrary and we compare
  \(\pa*{\varepsilon_{j,\infty}\pa*{\varepsilon_{i,j}(x)}}\) and
  \(\varepsilon_{i,\infty}(x)\) componentwise. So let \(j' : I\) be
  arbitrary. Since we are proving a proposition, we may assume to have \(k : I\)
  with \(j,j' \below k\) by semidirectedness of \(I\). We now calculate that
  \begingroup
  \allowdisplaybreaks
  \begin{align*}
    \pa*{\varepsilon_{j,\infty}\pa*{\varepsilon_{i,j}(x)}}_{j'}
    &\equiv \rho_{j,j'}\pa*{\varepsilon_{i,j}(x)} \\
    &= \pi_{j',k}\pa*{\varepsilon_{j,k}\pa*{\varepsilon_{i,j}(x)}}
    &&\text{(by \cref{rho-eq})} \\
    &= \pi_{j',k}\pa*{\varepsilon_{i,k}(x)}
    &&\text{(by \cref{epsilon-pi-comms})} \\
    &= \rho_{i,j'}(x)
    &&\text{(by \cref{rho-eq})} \\
    &\equiv \pa*{\varepsilon_{i,\infty}(x)}_{j'}
  \end{align*}
  \endgroup
  as desired.
\end{proof}

\begin{theorem}\label{limit}%
  The \(\V\)-dcpo \(D_\infty\) with the maps \(\pa*{\pi_{i,\infty}}_{i : I}\) is
  the limit of the diagram
  \(\pa*{\pa*{D_i}_{i : I} , \pa*{\pi_{i,j}}_{i \below j}}\).
  That is, given a \(\V \)-dcpo \(E : \DCPO{V}{U'}{T'}\) and Scott continuous
  functions \(f_i : E \to D_i\) for every \(i : I\) such that the diagram
  \begin{equation}\label{fs-are-cone}
    \begin{tikzcd}
      E \ar[dr,"f_j"'] \ar[rr,"f_i"] & & D_i \\
      & D_j \ar[ur,"\pi_{i,j}"']
    \end{tikzcd}
  \end{equation}
  commutes for every \(i \below j\),
  we have a unique Scott continuous function \(f_\infty : E \to D_\infty\) making
  the diagram
  \begin{equation}\label{f-infty}
    \begin{tikzcd}
      E \ar[dr,dashed,"f_\infty"'] \ar[rr,"f_i"] & & D_i \\
      & D_\infty \ar[ur,"\pi_{i,\infty}"']
    \end{tikzcd}
  \end{equation}
  commute for every \(i : I\).
\end{theorem}
\begin{proof}
  Note that \cref{f-infty} dictates that we must have
  \(\pa*{f_\infty(y)}_i = f_i(y)\) for every \(i : I\). Hence, we define
  \(f_\infty : E \to D_\infty\) as
  \(f_\infty(y) \colonequiv \lambdadot{i : I}{f_i(y)}\), which is Scott
  continuous because each \(f_i\) is and suprema are calculated pointwise in
  \(D_\infty\). To see that \(f_\infty\) is well-defined, i.e.\ that
  \(f_\infty(y)\) is indeed an element of \(D_\infty\), observe that for every
  \(i \below j\), the equation
  \(\pi_{i,j}\pa*{\pa*{f_\infty(y)}_j} \equiv \pi_{i,j}\pa*{f_j(y)} = f_i(y)\)
  holds because of \cref{fs-are-cone}.
\end{proof}

It should be noted that in the above universal property \(E\) can have its
carrier in any universe \(\U'\) and its order taking values in any universe
\(\T'\), even though we required all \(D_i\) to have their carriers and orders
in two fixed universes \(\U \) and \(\T \), respectively.

\begin{lemma}\label{epsilon-infty-fam-monotone}
  If \(i \below j\) in \(I\), then
  \(\varepsilon_{i,\infty}\pa{\sigma_i} \below
  \varepsilon_{j,\infty}\pa{\sigma_j}\) for every \(\sigma : D_\infty\).
\end{lemma}
\begin{proof}
  The order of \(D_\infty\) is pointwise, so we compare
  \(\varepsilon_{i,\infty}\pa{\sigma_i}\) and
  \(\varepsilon_{j,\infty}\pa{\sigma_j}\) at an arbitrary component \(k : I\).
  We may assume to have \(m : I\) such that \(j,k \below m\) by semidirectedness
  of \(I\). We then calculate that
  \begingroup
  \allowdisplaybreaks
  \begin{align*}
    \pa*{\varepsilon_{i,\infty}(\sigma_i)}_k
    &\equiv \rho_{i,k}(\sigma_i) \\
    &= \pa*{\pi_{k,m} \circ \varepsilon_{i,m}}(\sigma_i)
    &&\text{(by \cref{rho-eq})} \\
    &= \pa*{\pi_{k,m} \circ \varepsilon_{i,m} \circ \pi_{i,j}}(\sigma_j)
    &&\text{(since \(\sigma\) is an element of \(D_\infty\))} \\
    &= \pa*{\pi_{k,m} \circ \varepsilon_{j,m} \circ \varepsilon_{i,j} \circ\pi_{i,j}}(\sigma_j)
    &&\text{(by \cref{epsilon-pi-comms})} \\
    \shortintertext{But \(\varepsilon_{i,j} \circ\pi_{i,j}\) is deflationary
    and \(\pi_{k,m} \circ \varepsilon_{j,m}\) is monotone, so}
    &\below \pa*{\pi_{k,m} \circ \varepsilon_{j,m}}(\sigma_j) \\
    &= \rho_{j,k}(\sigma_j)
    &&\text{(by \cref{rho-eq})} \\
    &\equiv \pa*{\varepsilon_{j,\infty}(\sigma_j)}_k,
  \end{align*}
  \endgroup
  as we wished to show.
\end{proof}

\begin{lemma}\label{sigma-sup-of-epsilon-pis}
  Every element \(\sigma : D_\infty\) is equal to the directed supremum
  \(\bigsqcup_{i : I} \varepsilon_{i,\infty}\pa*{\sigma_i}\).
\end{lemma}
\begin{proof}
  The domain of the family is inhabited, because \((I,\below)\) is assumed to be
  directed.
  Moreover, if we have \(i,j : I\), then there exists \(k : I\) with
  \(i,j \below k\), which implies
  \(\varepsilon_{i,\infty}(\sigma_i),\varepsilon_{j,\infty}(\sigma_j) \below
  \varepsilon_{k,\infty}(\sigma_k)\) by \cref{epsilon-infty-fam-monotone}.
  Thus, the family \(i \mapsto \varepsilon_{i,\infty}(\sigma_i)\) is indeed
  directed.
  To see that its supremum is indeed \(\sigma\) we use antisymmetry at an
  arbitrary component \(j : I\).
  Firstly, observe that
  \begin{align*}
    \sigma_j
    &= \pi_{j,j}\pa*{\varepsilon_{j,j}(\sigma_j)} &&\text{(by \cref{epsilon-pi-id})} \\
    &= \rho_{j,j}(\sigma_j) &&\text{(by \cref{rho-eq})} \\
    &\equiv \pa*{\varepsilon_{j,\infty}(\sigma_j)}_j \\
    &\below \pa*{\textstyle\bigsqcup_{i : I}\varepsilon_{i,\infty}(\sigma_i)}_j
    &&\text{(since suprema are computed pointwise in \(D_\infty\))}.
  \end{align*}
  Secondly, to prove that
  \(\pa*{\bigsqcup_{i : I}\varepsilon_{i,\infty}(\sigma_i)}_j \below \sigma_j\)
  it suffices to show that
  \(\pa*{\varepsilon_{i,\infty}(\sigma_i)}_j \below \sigma_j\) for every
  \(i : I\). But this just says that
  \(\varepsilon_{i,\infty} \circ \pi_{i,\infty}\) is a deflation, which was
  proved in \cref{epsilon-pi-infty-ep-pair}.
\end{proof}

Although the composites \(\varepsilon_{i,\infty} \circ \pi_{i,\infty}\) are
deflations for each \(i : I\), the supremum of all of them is the identity. This
fact will come in useful in~\cref{sec:Scott-D-infty}.
\begin{lemma}\label{epsilon-pi-sup}
  The family \(i \mapsto \varepsilon_{i,\infty} \circ \pi_{i,\infty}\) is
  directed in the exponential \(D_\infty^{D_\infty}\) and its supremum is the
  identity on \(D_\infty\).
\end{lemma}
\begin{proof}
  The order and suprema are given pointwise in exponentials, so this follows
  from~\cref{sigma-sup-of-epsilon-pis}.
\end{proof}

\begin{theorem}\label{colimit}
  The \(\V\)-dcpo \(D_\infty\) with the maps
  \(\pa*{\varepsilon_{i,\infty}}_{i : I}\) is the colimit of the diagram
  \(\pa*{\pa*{D_i}_{i : I} , \pa*{\varepsilon_{i,j}}_{i \below j}}\).
  That is, given a \(\V\)-dcpo \(E : \DCPO{V}{U'}{T'}\) and Scott continuous
  functions \(g_i : D_i \to E\) for every \(i : I\) such that the diagram
  \begin{equation}\label{gs-are-cocone}
    \begin{tikzcd}
      D_i \ar[dr,"\varepsilon_{i,j}"'] \ar[rr,"g_i"] & & E \\
      & D_j \ar[ur,"g_j"']
    \end{tikzcd}
  \end{equation}
  commutes for every \(i \below j\),
  we have a unique Scott continuous function \(g_\infty : D_\infty \to E\) making
  the diagram
  \begin{equation}\label{g-infty}
    \begin{tikzcd}
      D_i \ar[dr,"\varepsilon_{i,\infty}"'] \ar[rr,"g_i"] & & E \\
      & D_\infty \ar[ur,dashed,"g_\infty"']
    \end{tikzcd}
  \end{equation}
  commute for every \(i : I\).
\end{theorem}
\begin{proof}
  Note that any such Scott continuous function \(g_\infty\) must satisfy
  \begin{align*}
    g_\infty(\sigma)
    &= g_\infty\pa*{\textstyle\bigsqcup_{i : I} \varepsilon_{i,\infty}(\sigma_i)}
    &&\text{(by \cref{sigma-sup-of-epsilon-pis})} \\
    &= \textstyle\bigsqcup_{i : I} g_\infty\pa*{\varepsilon_{i,\infty}(\sigma_I)}
    &&\text{(as \(g_\infty\) is assumed to be Scott continuous)} \\
    &= \textstyle\bigsqcup_{i : I}g_i(\sigma_i)
    &&\text{(by \cref{g-infty})}
  \end{align*}
  for every \(\sigma : D_\infty\).
  Accordingly, we define \(g_\infty\) by
  \(g_\infty(\sigma) \colonequiv \bigsqcup_{i : I}g_i(\sigma_i)\), where we
  verify that the family is indeed directed:
  If we have \(i,j : I\), then there exists \(k : I\) with \(i,j \below k\), and
  we have
  \begin{align*}
    g_i(\sigma_i)
    &= g_i\pa*{\pi_{i,k}(\sigma_k)}
    &&\text{(since \(\sigma\) is an element of \(D_\infty\))} \\
    &= g_k\pa*{\varepsilon_{i,k}\pa*{\pi_{i,k}(\sigma_k)}}
    &&\text{(by \cref{gs-are-cocone})} \\
    &\below g_k(\sigma_k)
    &&\text{(since \(\varepsilon_{i,k} \circ \pi_{i,k}\) is deflationary and \(g_k\) is monotone)},
  \end{align*}
  and similarly, \(g_j(\sigma_j) \below g_k(\sigma_k)\).
  To see that \(g_\infty\) satisfies \cref{g-infty}, let \(x : D_i\) be
  arbitrary and first observe that
  \[
    g_\infty\pa*{\varepsilon_{i,\infty}(x)} \equiv \textstyle\bigsqcup_{j :
      I}g_j\pa*{\pa*{\varepsilon_{i,\infty}(x)}_j} \equiv \textstyle\bigsqcup_{j
      : I}g_j\pa*{\rho_{i,j}(x)}.
  \]
  We now use antisymmetry to prove that this is equal to \(g_i(x)\).
  In one direction this is easy as
  \( g_i(x) = \pa*{g_i \circ \pi_{i,i} \circ \varepsilon_{i,i}}(x) \equiv
  g_i\pa*{\rho_{i,i}(x)} \below \textstyle\bigsqcup_{j :
    I}g_j\pa*{\rho_{i,j}(x)} \).  In the other direction, it suffices to prove
  that \(g_j(\rho_{i,j}(x)) \below g_i(x)\) for every \(j : I\). By directedness
  of \(I\) there exists \(k : I\) with \(i,j \below k\) so that
  \begin{align*}
    g_j(\rho_{i,j}(x))
    &= \pa*{g_j \circ \pi_{j,k} \circ \varepsilon_{i,k}}(x)
    &&\text{(by \cref{rho-eq})} \\
    &= \pa*{g_k \circ \varepsilon_{j,k} \circ \pi_{j,k} \circ \varepsilon_{i,k}}(x)
    &&\text{(by \cref{gs-are-cocone})} \\
    \shortintertext{But \(\varepsilon_{j,k} \circ \pi_{j,k}\) is deflationary and \(g_k\) is monotone, so}
    &\below \pa*{g_k \circ \varepsilon_{i,k}}(x) \\
    &= g_i(x)
    &&\text{(by \cref{gs-are-cocone})},
  \end{align*}
  as we wished to show.

  Finally, we verify that \(g_\infty\) is Scott continuous. We first check that
  \(g_\infty\) is monotone. If \(\sigma \below \tau\) in \(D_\infty\), then
  \(g_\infty(\sigma) \equiv \bigsqcup_{i : I} g_i(\sigma_i) \below \bigsqcup_{i
    : I} g_i(\tau_i) \equiv g_\infty(\tau)\), as each \(g_i\) is monotone.
  It remains to show that
  \(g_\infty\pa*{\bigsqcup \alpha} \below \bigsqcup \pa*{g_\infty \circ
    \alpha}\) for every directed family \(\alpha : A \to D_\infty\).
  By definition of \(g_\infty\), it suffices to show that
  \(g_i\pa*{\pa*{\bigsqcup\alpha}_i} \below \bigsqcup \pa*{g_\infty \circ
    \alpha}\) for every \(i : I\).
  By continuity of \(g_i\) it is enough to establish that
  \(g_i\pa*{\pa*{\alpha(a)}_i} \below \bigsqcup \pa*{g_\infty \circ \alpha}\)
  for every \(a : A\). But this holds as
  \( g_i\pa*{\pa*{\alpha(a)}_i} \below g_\infty(\alpha(a)) \below \bigsqcup
  \pa*{g_\infty \circ \alpha} \), completing our proof.
\end{proof}

\begin{proposition}\label{locally-small-bilimit}
  The bilimit of locally small dcpos is locally small, i.e.\ if every \(\V\)-dcpo
  \(D_i\) is locally small for all \(i : I\), then so is \(D_\infty\).
\end{proposition}
\begin{proof}
  If every \(D_i\) is locally small, then for every \(i : I\), we have a
  \emph{specified} \(\V\)-valued partial order \(\below_{\V}^i\) on \(D_i\) such
  that for every \(i : I\) and every \(x,y : D_i\), we have an equivalence
  \(\pa{x \below_{D_i} y} \simeq \pa{x \below_{\V}^i y}\).
  Hence,
  \(\pa{\sigma \below_{D_\infty} \tau} \equiv \pa{\Pi_{i : I}\pa{\sigma_i
      \below_{D_i} \tau_i}} \simeq \pa{\Pi_{i : I}\pa{\sigma_i \below_{\V}^i
      \tau_i}}\), but the latter is small, because \(I : \V\) and
  \(\below_{\V}^i\) is \(\V\)-valued.
\end{proof}

\section{Scott's \texorpdfstring{\(D_\infty\)}{D-infinity} model of the untyped
  \texorpdfstring{\(\lambda\)}{lambda}-calculus}\label{sec:Scott-D-infty}

We construct Scott's \(D_\infty\)~\cite{Scott1972} predicatively. Formulated
precisely, we construct a pointed \(D_\infty : \DCPOnum{0}{1}{1}\) such that
\(D_\infty\) is isomorphic to its self-exponential~\(D_\infty^{D_\infty}\),
employing the machinery from~\cref{sec:bilimits}.

\begin{definition}[\(D_n\)]
  We inductively define pointed dcpos \(D_n : \DCPOnum{0}{1}{1}\) for every
  natural number \(n\) by setting
  \(D_0 \colonequiv \lifting_{\U_0}\pa*{\One_{\U_0}}\) and
  \(D_{n+1} \colonequiv D_n^{D_n}\).
\end{definition}

In light of~\cref{universe-levels-of-lifting-and-exponentials} we highlight the
fact that every \(D_n\) is a \(\U_0\)-dcpo with carrier in \(\U_1\) by the
discussion of universe parameters of exponentials
in~\cref{exponential-universe-parameters}.%

\begin{definition}[\(\varepsilon_n\), \(\pi_n\)]
  We inductively define for every natural number \(n\), two Scott continuous maps
  \(\varepsilon_n : D_n \to D_{n+1}\) and \({\pi_n : D_{n+1} \to D_n}\):
  \begin{enumerate}[(i)]
  \item
    \begin{itemize}
    \item \(\varepsilon_0 : D_0 \to D_1\) is given by mapping \(x : D_0\) to the
      continuous function that is constantly~\(x\),
    \item \(\pi_0 : D_1 \to D_0\) is given by evaluating a continuous function
      \(f : D_0 \to D_0\) at~\(\bot\) which is itself continuous
      by~\cref{exponential-universal-property},
    \end{itemize}
  \item
    \begin{itemize}
    \item \(\varepsilon_{n+1} : D_{n+1} \to D_{n+2}\) takes a continuous
      function \(f : D_n \to D_n\) to the continuous composite
      \(D_{n+1} \xrightarrow{\pi_n} D_n \xrightarrow{f} D_n
      \xrightarrow{\varepsilon_n} D_{n+1}\), and
    \item \(\pi_{n+1} : D_{n+2} \to D_{n+1}\) takes a continuous function
      \(f : D_{n+1} \to D_{n+1}\) to the continuous composite
      \(D_n \xrightarrow{\varepsilon_n} D_{n+1} \xrightarrow{f} D_{n+1}
      \xrightarrow{\pi_n} D_n\). \qedhere
    \end{itemize}
  \end{enumerate}
\end{definition}

\begin{lemma}
  The maps \(\pa*{\varepsilon_n,\pi_n}\) form an embedding-projection pair for
  every natural number \(n\).%
\end{lemma}
\begin{proof}
  We prove this by induction on \(n\). For \(n \equiv 0\) and arbitrary
  \(x : D_0\), we have
  \[
    \pi_0(\varepsilon_0(x)) \equiv \pi_0(\const_x) \equiv \const_x(\bot) \equiv
    x,
  \]
  so \(\varepsilon_0\) is indeed a section of \(\pi_0\).
  Moreover, for arbitrary \(f : D_1\), we have
  \[
    \varepsilon_0(\pi_0(f)) \equiv \varepsilon_0(f(\bot)) \equiv \const_{f(\bot)},
  \]
  so that for arbitrary \(x : D_0\) we get
  \(\pa*{\varepsilon_0(\pi_0(f))}(x) \equiv f(\bot) \below f(x)\) by
  monotonicity of~\(f\), proving that \(\varepsilon_0 \circ \pi_0\) is
  deflationary.

  Now suppose that the result holds for a natural number \(n\); we prove it for
  \(n+1\). For arbitrary \(f : D_n \to D_n\), we calculate that
  \[
    \pi_{n+1}\pa*{\varepsilon_{n+1}(f)}
    \equiv \pi_n \circ \varepsilon_{n+1}(f) \circ \varepsilon_n
    \equiv \pi_n \circ \varepsilon_n \circ f \circ \pi_n \circ \varepsilon_n
    = f,
  \]
  as \(\varepsilon_n\) is a section of \(\pi_n\) by induction hypothesis.
  The proof that \(\varepsilon_{n+1} \circ \pi_{n+1}\) is a deflation is
  similar.
\end{proof}

In order to apply the tools from \cref{sec:bilimits}, we will need
embedding-projection pairs \(\pa*{\varepsilon_{n,m},\pi_{n,m}}\) from \(D_n\) to
\(D_m\) whenever \(n \leq m\).

\begin{definition}[\(\varepsilon_{n,m},\pi_{n,m}\)]
  We define Scott continuous maps \(\varepsilon_{n,m} : D_n \to D_m\) and
  \({\pi_{n,m} : D_m \to D_n}\) for every two natural numbers \(n \leq m\) as
  follows:
  \begin{enumerate}[(i)]
  \item \(\varepsilon_{n,m}\) and \(\pi_{n,m}\) are both defined to be the
    identity if \(n = m\);
  \item if \(n < m\), then we define \(\varepsilon_{n,m}\) as the composite
    \[
      D_n \xrightarrow{\varepsilon_n} D_{n+1} \to \cdots \to D_{m-1}
      \xrightarrow{\varepsilon_{m-1}} D_m
    \]
    and \(\pi_{n,m}\) as the composite
    \[
      D_m \xrightarrow{\pi_m} D_{m-1} \to \cdots \to D_{n+1}
      \xrightarrow{\pi_{n}} D_n
    \]
  \end{enumerate}
  which yields embedding-projection pairs as they are compositions of them.
\end{definition}

Instantiating the framework of~\cref{sec:bilimits} with the above diagram of
objects \(D_n : \DCPOnum{0}{1}{1}\), we arrive at the construction of
\(D_\infty\) and appropriate embedding-projection pairs. Observe that
\(D_\infty\) is a \(\U_0\)-dcpo with carrier and order taking values
in~\(\U_1\), just like each \(D_n\), as was also mentioned
in~\cref{bilimit-universe-parameters}.

\begin{definition}[\(D_\infty\)]
  Applying \cref{def:D-infty,pi-infty,epsilon-infty} to the above diagram yields
  \(D_\infty : \DCPOnum{0}{1}{1}\) with embedding-projection pairs
  \(\pa*{\varepsilon_{n,\infty},\pi_{n,\infty}}\) from \(D_n\) to \(D_\infty\)
  for every natural number \(n\).
\end{definition}

\begin{lemma}\label{pi-is-strict}
  The function \(\pi_n : D_{n+1} \to D_n\) is strict for every natural number
  \(n\). Hence, so is \(\pi_{n,m}\) whenever \(n \leq m\).
\end{lemma}
\begin{proof}
  Both statements are proved by induction.
\end{proof}

\begin{proposition}
  The dcpo \(D_\infty\) is pointed.
\end{proposition}
\begin{proof}
  Since every \(D_n\) is pointed, we can consider the function
  \(\sigma : \Pi_{n : \Nat} D_n\) given by \(\sigma(n) \colonequiv
  \bot_{D_n}\). Then \(\sigma\) is an element of \(D_{\infty}\) by
  \cref{pi-is-strict} and it is the least, so \(D_\infty\) is indeed pointed.
\end{proof}

We now work towards showing that \(D_\infty\) is isomorphic to the exponential
\(D_\infty^{D_\infty}\). Note that this exponential is again an element of
\(\DCPOnum{0}{1}{1}\) by~\cref{exponential-universe-parameters}, so the universe
parameters do not increase.

\begin{definition}[\(\Phi_n\)]
  For every natural number \(n\), we define the continuous maps
  \begin{align*}
    \Phi_{n+1} : D_{n+1} &\to  D_\infty^{D_{\infty}} \\
    f &\mapsto \pa{D_\infty \xrightarrow{\pi_{n,\infty}} D_n
    \xrightarrow{f} D_n \xrightarrow{\varepsilon_{n,\infty}} D_\infty}
  \end{align*}
  and \(\Phi_0 : D_0 \to D_{\infty}^{D_{\infty}}\) as
  \(\Phi_1 \circ \varepsilon_0\).
\end{definition}

\begin{lemma}\label{Phi-commutes-with-epsilons}
  For every two natural numbers \(n \leq m\), the diagram

  \[
    \begin{tikzcd}
      D_n \ar[dr,"\varepsilon_{n,m}"'] \ar[rr,"\Phi_n"]
      & & D_\infty^{D_\infty} \\
      & D_m \ar[ur,"\Phi_m"']
    \end{tikzcd}
  \]
  commutes.
\end{lemma}
\begin{proof}
  By induction on the difference of the two natural numbers, it suffices to
  prove that for every natural number \(n\), the diagram
  \[
    \begin{tikzcd}
      D_n \ar[dr,"\varepsilon_n"'] \ar[rr,"\Phi_n"]
      & & D_\infty^{D_\infty} \\
      & D_{n+1} \ar[ur,"\Phi_{n+1}"']
    \end{tikzcd}
  \]
  commutes. But this follows from~\cref{epsilon-pi-infty-commutes} and unfolding
  the definition of \(\Phi_n\).
\end{proof}

\begin{definition}[\(\Phi\)]
  The map \(\Phi : D_\infty \to D_\infty^{D_\infty}\) is defined as the unique
  Scott continuous map induced by the \(\Phi_n\) via~\cref{colimit}.
\end{definition}

\begin{lemma}\label{Phi-alt}
  For \(\sigma : D_\infty\) we have %
  \(\Phi(\sigma) = \textstyle\bigsqcup_{n : \Nat}\Phi_{n+1}\pa*{\sigma_{n+1}}\).
\end{lemma}
\begin{proof}
  Recalling the proof of~\cref{colimit} we have
  \(\Phi(\sigma) \equiv \textstyle\bigsqcup_{n : \Nat}\Phi_{n}\pa*{\sigma_{n}}\),
  from which the claim follows easily.
\end{proof}

We now define a map in the other direction using that \(D_\infty\) is also the
limit.

\begin{definition}[\(\Psi_n\)]
  For every natural number \(n\), we define the continuous maps
  \begin{align*}
    \Psi_{n+1} : D_\infty^{D_\infty} &\to D_{n+1}  \\
    f &\mapsto \pa{D_n \xrightarrow{\varepsilon_{n,\infty}} D_\infty
    \xrightarrow{f} D_\infty \xrightarrow{\pi_{n,\infty}} D_n}
  \end{align*}
  and \(\Psi_0 : D_\infty^{D_\infty} \to D_0\) as \(\pi_0 \circ \Psi_1\).
\end{definition}

\begin{lemma}\label{Psi-commutes-with-pis}
  For every two natural numbers \(n \leq m\), the diagram
  \[
    \begin{tikzcd}
      D_\infty^{D_\infty} \ar[dr,"\Psi_m"'] \ar[rr,"\Psi_n"]
      & & D_n \\
      & D_m \ar[ur,"\pi_{n,m}"']
    \end{tikzcd}
  \]
  commutes.
\end{lemma}
\begin{proof}
  Similar to \cref{Phi-commutes-with-epsilons}.
\end{proof}

\begin{definition}[\(\Psi\)]
  The map \(\Psi : D_\infty^{D_\infty} \to D_\infty\) is defined as the unique
  Scott continuous map induced by the \(\Psi_n\) via~\cref{limit}.
\end{definition}

\begin{lemma}\label{Psi-alt}
  For \(f : D_\infty^{D_\infty}\) we have
  \(\Psi(f) = \textstyle\bigsqcup_{n : \Nat}\varepsilon_{n+1,\infty}\pa*{\Psi_{n+1}(f)}\).
\end{lemma}
\begin{proof}
  Notice that
  \begin{align*}
    \Psi(f)
    &= \textstyle\bigsqcup_{n : \Nat}\varepsilon_{n,\infty}\pa*{\pi_{n,\infty}\pa*{\Psi(f)}}
    &&\text{(by \cref{sigma-sup-of-epsilon-pis})} \\
    &= \textstyle\bigsqcup_{n : \Nat}\varepsilon_{n,\infty}\pa*{\Psi_n(f)}
    &&\text{(by \cref{f-infty})},
  \end{align*}
  from which the claim follows easily.
\end{proof}

\begin{theorem}\label{isomorphic-to-self-exponential}%
  The maps \(\Phi\) and \(\Psi\) are inverses and hence, \(D_\infty\) is
  isomorphic to \(D_\infty^{D_\infty}\).
\end{theorem}
\begin{proof}
  We first recall that
  \begin{equation*}\label{epsilon-section-of-pi}
    \varepsilon_{n,\infty} \text{ is a section of } \pi_{n,\infty} \text{ for every } n \in \Nat.
    \tag{\(\ast\)}
  \end{equation*}
  For arbitrary \(\sigma : D_\infty\) we then calculate that
  \begin{align*}
    \Psi(\Phi(\sigma))
    &= \Psi\pa*{\textstyle\bigsqcup_{n : \Nat}\Phi_{n+1}\pa*{\sigma_{n+1}}}
    &&\text{(by \cref{Phi-alt})} \\
    &= \textstyle\bigsqcup_{n : \Nat}\Psi\pa*{\Phi_{n+1}\pa*{\sigma_{n+1}}}
    &&\text{(by continuity of \(\Psi\))} \\
    &= \textstyle\bigsqcup_{n : \Nat}\textstyle\bigsqcup_{m : \Nat}
      \varepsilon_{m+1,\infty}\pa*{\Psi_{m+1}\pa*{\Phi_{n+1}\pa*{\sigma_{n+1}}}}
    &&\text{(by \cref{Psi-alt})} \\
    &= \textstyle\bigsqcup_{n : \Nat}\varepsilon_{n+1,\infty}\pa*{\Psi_{n+1}\pa*{\Phi_{n+1}\pa*{\sigma_{n+1}}}} \\
    &\equiv \textstyle\bigsqcup_{n : \Nat}\varepsilon_{n+1,\infty}\pa*{\pi_{n,\infty} \circ \varepsilon_{n,\infty} \circ \sigma_{n+1} \circ \pi_{n,\infty} \circ \varepsilon_{n,\infty}}
    &&\text{(by definition)} \\
    &= \textstyle\bigsqcup_{n : \Nat} \varepsilon_{n+1,\infty}\pa*{\sigma_{n+1}}
    &&\text{(by \eqref{epsilon-section-of-pi})} \\
    &= \sigma
    &&\text{(by \cref{sigma-sup-of-epsilon-pis})},
  \end{align*}
  so \(\Phi\) is indeed a section of \(\Psi\).
  Moreover, for arbitrary \(f : D_\infty^{D_\infty}\) we calculate that
  \begin{align*}
    \Phi\pa*{\Psi(f)}
    &= \Phi\pa*{\textstyle\bigsqcup_{n : \Nat}\varepsilon_{n+1,\infty}\pa*{\Psi_{n+1}(f)}}
    &&\text{(by \cref{Psi-alt})} \\
    &= \textstyle\bigsqcup_{n : \Nat}\Phi\pa*{\varepsilon_{n+1,\infty}\pa*{\Psi_{n+1}(f)}}
    &&\text{(by continuity of \(\Phi\))} \\
    &= \textstyle\bigsqcup_{n : \Nat}\textstyle\bigsqcup_{m : \Nat}
      \Phi_{m+1}\pa*{\pi_{m+1,\infty}\pa*{\varepsilon_{n+1,\infty}\pa*{\Psi_{n+1}(f)}}}
    &&\text{(by \cref{Phi-alt})} \\
    &= \textstyle\bigsqcup_{n : \Nat}\Phi_{n+1}\pa*{\pi_{n+1,\infty}\pa*{\varepsilon_{n+1,\infty}\pa*{\Psi_{n+1}(f)}}} \\
    &= \textstyle\bigsqcup_{n : \Nat}\Phi_{n+1}\pa*{\Psi_{n+1}(f)}
    &&\text{(by \eqref{epsilon-section-of-pi})} \\
    &\equiv \textstyle\bigsqcup_{n : \Nat}\pa*{\varepsilon_{n,\infty} \circ \pi_{n,\infty} \circ f \circ \varepsilon_{n,\infty} \circ \pi_{n,\infty}}
    &&\text{(by definition)} \\
    &= \pa*{\textstyle\bigsqcup_{n : \Nat}\varepsilon_{n,\infty} \circ \pi_{n,\infty}} \circ f \circ \pa*{\textstyle\bigsqcup_{m : \Nat}\varepsilon_{m,\infty} \circ \pi_{m,\infty}} \\
    &= f
    &&\text{(by \cref{epsilon-pi-sup})},
  \end{align*}
  finishing the proof.
\end{proof}

\begin{remark}
  Of course, \cref{isomorphic-to-self-exponential} is only interesting when
  \(D_\infty\) is not the trivial one-element dcpo. Fortunately, \(D_\infty\)
  has (infinitely) many elements besides the least
  element\(\bot_{D_\infty}\). For instance, we can consider
  \(x \colonequiv \eta(\star) : D_0\) and observe that
  \(\varepsilon_{0,\infty}(x)\) is an element of \(D_\infty\) not equal
  to~\(\bot_{D_\infty}\), since \(x \neq \bot_{D_0}\).
\end{remark}

\section{Concluding remarks}\label{sec:conclusion}
We have shown how to develop basic domain theory in constructive and predicative
univalent foundations with type universes playing a fundamental role in keeping
track of relative sizes.
We highlighted the use of the propositional truncation and illustrated our
development with a predicative reconstruction of Scott's \(D_\infty\).

A priori it is not clear that \(D_\infty\) should exist in predicative univalent
foundations and it is one of the contributions of this work that this is indeed
possible.
Our construction largely follows the classical development of Scott's original
paper~\cite{Scott1972}, but with some crucial differences.
First of all, we carefully keep track of the universe parameters and try to be
as general as possible. In the particular case of Scott's \(D_\infty\) model of
the untyped \(\lambda\)-calculus, we obtain a \(\U_0\)-dcpo whose carrier lives
in the second universe \(\U_1\).
The ability of the proof assistant \Agda\ to infer universe levels has been
invaluable in this regard.
Secondly, difference arises from proof relevance and these complications are
tackled with techniques in univalent foundations
and~\cite[Theorem~5.4]{KrausEtAl2017} in particular, as discussed right
before~\cref{kappa-is-constant}, for example.
Finally, we generalised Scott's treatment from sequential bilimits to directed
bilimits.

In the Part~II paper~\cite{deJongEscardoCompanion}, we further show that
\(D_\infty\) is an algebraic dcpo with a small compact basis. More generally, we
take inspiration from category theory to present a predicative theory of
continuous and algebraic domains.

\section{Acknowledgements}
I am grateful to Andrej Bauer and Dana Scott as they provided the spark for this
paper by (independently) asking me whether we could have \(D_\infty\) in
univalent foundations.
I also thank Andrej and Vincent Rahli for valuable feedback on my thesis.
Finally, I am very thankful to Mart\'in Escard\'o with whom I had many fruitful
discussions on the broader topic of developing domain theory constructively and
predicatively in univalent foundations.

This work was supported by Cambridge Quantum Computing and Ilyas Khan
[Dissertation Fellowship in Homotopy Type Theory]; and The Royal Society [grant
reference URF{\textbackslash}R1{\textbackslash}191055].

\setcounter{biburllcpenalty}{8000}
\setcounter{biburlucpenalty}{7000}
\setcounter{biburlnumpenalty}{7000}
\printbibliography

\end{document}